%% file: main.tex
\documentclass[11pt]{article}
\usepackage[margin=1in]{geometry}
\usepackage{times}
\usepackage{amsfonts}
\usepackage{amsmath,amsthm,amssymb}
\usepackage{latexsym}
\usepackage{epic}
\usepackage{epsfig}
\usepackage{hyperref}
\usepackage{verbatim}
\usepackage[justification=centering]{caption}
\usepackage{enumitem}
\usepackage{color}
\allowdisplaybreaks[1]
\usepackage[numbers]{natbib}
\usepackage{algorithm}
\usepackage{algorithmic}
\usepackage{setspace}
\usepackage{booktabs}

\input{macros}

\title{Algorithmic Persuasion with  No Externalities}

\author{
	Shaddin Dughmi\thanks{Supported in part by NSF CAREER Award CCF-1350900.} \\
	Department of Computer Science\\
	University of Southern California\\
	{\tt shaddin@usc.edu}
	\and
	Haifeng Xu\thanks{Supported by NSF grant CCF-1350900.} \\
	Department of Computer Science\\
	University of Southern California\\
	{\tt haifengx@usc.edu}
}

\date{}

\begin{document}

\maketitle
\begin{abstract} 
	We study the algorithmics of information structure design --- a.k.a. persuasion or signaling --- in a fundamental special case introduced by Arieli and Babichenko: multiple agents, binary actions,  and no inter-agent externalities. Unlike prior work on this model, we allow many states of nature.  We assume that the principal's objective is  a monotone set function, and study the problem both in the public signal and private signal models, drawing a sharp contrast between the two  in terms of both efficacy and computational complexity.
	
	When private signals are allowed, our results are largely positive and quite general. First, we use linear programming duality and the equivalence of separation and optimization to show polynomial-time equivalence between (exactly) optimal signaling and the problem of maximizing the objective function plus an additive function. This yields an efficient implementation of the optimal scheme when the objective is supermodular or anonymous. Second, we exhibit a $(1-1/e)$-approximation of the optimal private signaling scheme, modulo an additive loss of $\eps$, when the objective function is submodular. These two results simplify, unify, and generalize results of \cite{Arieli2016,babichenko2016}, extending them from a binary state of nature to many states (modulo the additive loss in the latter result). Third, we consider the binary-state case with a submodular objective, and simplify and slightly strengthen the result of  \cite{babichenko2016}  to obtain a $(1-1/e)$-approximation via a scheme which (i)  signals independently to each receiver and (ii) is ``oblivious'' in that it does \emph{not} depend on the objective function so long as it is monotone submodular.
	
	When only a public signal is allowed, our results are negative. First, we show that it is NP-hard to approximate the optimal public scheme, within any constant factor, even when the objective is additive. Second, we show that the optimal private scheme can outperform the optimal public scheme, in terms of maximizing the sender's objective, by a polynomial factor.
	
\end{abstract}

\section{Introduction}

%Persuasion, defined as the act of exploiting an informational advantage in order to influence the decisions of others, is of substantial economic importance. As emphasized by McCloskey and Klamer \cite{McCloskey95}, persuasive interactions comprise a considerable share of economic activities. The fundamental \emph{Bayesian persuasion} model \cite{Kamenica2011}  captures  basic settings where a principle, call \emph{her} the \emph{sender}, tries to persuade a decision maker, call \emph{him} the \emph{receiver}, to take an action that is more favorable to the sender. However, in many domains, a sender needs to persuade \emph{a group of receivers} and the sender's payoff function also depends on all the receivers' actions. For example,  a politician may want to persuade voters in an election, an advertiser may want to persuade  potential customers during marketing, and an auctioneer may want to persuade bidders in an auction etc. In these settings, the sender needs to communicate with different receivers in order to optimize a global objective that could depend on all the receivers' actions. 

Information structure design studies how beliefs influence behavior, both of individuals and of groups, and how to shape those beliefs in order to achieve desired outcomes. The key object of study here is the \emph{information structure}, which describes ``who knows what'' about the parameters governing the payoff function of a game of incomplete information. These parameters, collectively termed the \emph{state of nature}, encode uncertainty in the environment, and their prior distribution together with the information structure determine agents' equilibrium behavior. 

When studied descriptively, information structures lend  structural insight into the space of potential equilibria of a game of incomplete information (termed Bayes correlated equilibria in \cite{Bergemann16Bayes}), and can provide answers to comparative statics questions regarding the provision or withholding of information. The associated prescriptive questions are fundamentally algorithmic, and are often termed \emph{persuasion} or \emph{signaling}. Persuasion is the task faced by a principal --- call her the \emph{sender} --- who has privileged access to the state of nature, and can selectively provide information to the agent(s) in the game --- call him/them the \emph{receiver(s)} --- in order to further her own objectives. Such a sender effectively \emph{designs} the information structure, and faces the algorithmic problem of selecting which of her information to reveal, to whom, and whether and how to add noise to this information.  As in much of the prior work in this area, when we think of the information structure as an algorithm implemented by the sender, we call it a \emph{signaling scheme}.

Both descriptive and prescriptive questions in this space have enjoyed an explosion of research interest of late, in particular at the intersection of economics and computer science.  Our discussion can not do justice to the large literature in this area, including the wealth of work on related models. We therefore refer the curious reader to the recent survey by \cite{Dughmi2017} and the references therein. The model we focus on in this paper, and for which we provide a thorough exploration through the algorithmic lens, is the fundamental special case of the multi-agent persuasion problem when restricted to binary actions and no externalities.  
%\sddelete{We discuss this model next, before segueing into a discussion of related work, followed by our results.}

\subsection*{Our Model and its Significance}

Our technical and conceptual starting point, and arguably the most influential model in this area, is the \emph{Bayesian persuasion} model of \cite{Kamenica2011}. Here the sender faces a single receiver, and her task is to persuade the receiver to take an action which is more favored by the sender.  As in persuasion more generally, the sender's ``leverage'' is her informational advantage --- namely, access to the realization of the state of nature, which determines the payoffs of the actions to both the receiver and the sender. The receiver, in contrast, a-priori knows nothing about the state of nature besides its prior distribution (the common prior). \citet{Kamenica2011} assume that the sender can \emph{commit} to her policy for revealing the information prior to realization of the state of nature, and study how and when this signaling scheme can increase the sender's utility. 

The model we use is one recently proposed by \citet{Arieli2016}, both generalizing and restricting  aspects of the model of \cite{Kamenica2011}.
Here the sender interacts with multiple receivers, each of whom is restricted to a binary choice of actions. Without loss, we denote these actions by $1$ and $0$. 
As mentioned in \cite{Arieli2016,babichenko2016}, settings like this arise when a manager seeks to persuade investors to invest in a project, or when a principal persuades opinion leaders in a social network with the goal of maximizing social influence.
Each receiver's utility depends only on his own action and the state of nature, but crucially not on the actions of other receivers --- the \emph{no externality} assumption. The sender's utility, on the other hand, depends on the state of nature as well as the \emph{profile} of receiver actions --- since actions are binary, the sender's utility can be viewed as a \emph{set function} on receivers. 
As in \cite{Kamenica2011}, the state of nature is drawn from a common prior, and the sender can commit to a policy of revealing information regarding the realization of the state of nature. Since there are multiple receivers, this policy --- the information structure --- is more intricate, since it can reveal different, and correlated, information to different receivers. As made clear in \cite{Arieli2016}, such flexibility is crucial to the sender unless receivers are homogeneous and the sender's utility function highly structured (for example, additively separable across receivers). In particular, if restricted to a \emph{public communication channel}, the sender is limited in her ability to discriminate between receivers and correlate their actions, whereas a \emph{private communication channel} provides more flexibility.% As made clear in \cite{Arieli2016}, when the sender's utility is not additively separable across receivers, such differentiation and correlations are important and the sender's problem does not  decompose into individual  Bayesian persuasion problems.

Our results, and the models of \citet{Kamenica2011} and \citet{Arieli2016}, are crucially underlied by the assumption that the sender has the power of commitment to a signaling scheme. The commitment assumption is not as unrealistic as it might first sound, and a number of arguments to that effect are provided in \cite{Rayo10,Kamenica2011, Dughmi2017}. We mention one of those arguments here: commitment  arises organically at equilibrium if the sender and receiver(s) interact repeatedly over a long horizon, in which case commitment can be thought of as a proxy for ``establishing credibility.'' 

If one permits commitment by the sender, restricts attention to a common prior distribution, and postulates that receiver(s) are not privately informed through channels not controlled by the sender, then the model of \cite{Kamenica2011} is \emph{the} special case of persuasion for a single receiver, and  the model of~\cite{Arieli2016} is \emph{the} special case of persuasion for multiple receivers with binary actions and no externalities. Viewed this way, both are fundamental special cases of the general information structure design problem, meaningfully restricted. It is thus only fitting that they be thoroughly explored via an algorithmic lens, en-route to a more general understanding of the algorithmics of information in multi-agent settings. A computational study of the single-receiver Bayesian persuasion model of \cite{Kamenica2011} was undertaken by  \citet{Dughmi2016}. As for multi-agent persuasion with binary actions and no externalities, a partial algorithmic understanding is provided by \cite{Arieli2016,babichenko2016} in the special case of a binary state of nature, and this paper picks up there.

\subsection*{Context: Private and Public  Persuasion}
At their most general, signaling schemes  can reveal different information to different receivers, through \emph{private} communication channels. Such schemes play a dual role: they \emph{inform} receivers, possibly asymmetrically, and they \emph{coordinate} their behavior by correlating the information provided. In some settings, however, such private communication channels are unrealistic, and the sender is constrained to a \emph{public} communication channel. Work on both descriptive and prescriptive questions regarding information structure design in multi-agent settings can be classified along these lines.  The extent to which a public communication channel limits the sender's powers of persuasion is a fundamental question which has not been thoroughly explored.

Much of the earlier work on information structure design, in particular its computational aspects, focused on public signaling models. This includes work on signaling in auctions \cite{Emek12,Miltersen12,Guo13,Dughmi14b}, voting \cite{Alonso14}, routing~\cite{Bhaskar2016}, and  abstract game models~\cite{Dughmi14a,mixture_selection,Bhaskar2016,rubinstein2015eth}. The work of \cite{mixture_selection} is relevant to this paper,  in that they identify conditions under which public persuasion problems are tractable to approximate, and prove impossibility results in some cases where those conditions are violated. Our hardness proof in Section~\ref{sec:public} is in part based on some of their ideas.
% in that they identify conditions under which public persuasion problems are tractable to approximate, and prove impossibility results in some cases where those conditions are violated. Our hardness proof in Section~\ref{sec:public} is in part based on some of their ideas.

%If the sender is constrained to reveal the same information to all receivers, she can only send a signal publicly to receivers via a \emph{public signaling scheme}. After receiving the signal, each player performs a Bayesian update and reaches a common posterior belief of the state of nature. The design of optimal public signaling has been examined in the setting of voting \cite{Alonso14,mixture_selection}.  Alonso and C\^{a}mara \cite{Alonso14} characterize several properties of the optimal public signaling scheme. However, they do not derive an algorithm for the optimal scheme.  Cheng et al. \cite{mixture_selection} present a bi-criteria PTAS for the optimal public signaling scheme. To our knowledge, there has not been much work on computing the optimal public signaling schemes for more general sender objectives. \sdmargincomment{Isn't ``general objective'' the point of the mixture selection paper? I'm not sure I understand the intent of this sentence} We believe that our computational hardness results for optimal public signaling partially explain this slow progress. 

Private persuasion has been less thoroughly explored, particularly through the computational lens. The space of (private channel) information structures  is studied by \citet{Bergemann16Bayes}, who observe that these information structures and their associated equilibria form a generalization of correlated equilibria, and term the generalization the \emph{Bayes Correlated Equilibrium (BCE)}. The space of BCEs is characterized in two-agent two-action games by \cite{Taneva2015}. Moreover, recent work explores private persuasion in the context of voting \cite{Wang2015,Chan2016,Bardhi2016} and network routing \cite{Vasserman2015}. 

\citet{Arieli2016} pose the model we consider in this paper, with the goal of studying private persuasion. We believe that their model wisely simplifies the general multi-agent persuasion problem by removing the most thorny aspect limiting progress in the study of private schemes: externalities. Specifically, by assuming that receivers only influence each others' payoffs through the choices of the sender, they sidestep the equilibrium selection and computation concerns which would arise in more general settings. Assuming a binary state of nature, \citet{Arieli2016} provide explicit characterizations of the optimal private signaling scheme for three of the most natural classes of sender objective functions: viewing the sender's objective as a set function on receivers, those are supermodular,  anonymous monotone submodular, and supermajority functions. In all three cases, their characterization can be easily converted to an efficient algorithm. Moreover \citet{Arieli2016} provide necessary and sufficient conditions under which public signaling schemes match the performance of private signaling schemes.  In follow-up work, \citet{babichenko2016} also consider the binary state setting, and reduce the private persuasion problem to the problem of computing the concave closure of the sender's utility function. This connection yields a  $(1-\frac{1}{e}-\epsilon)$-approximate private signaling scheme for  monotone submodular utility functions, and an optimal private scheme for  anonymous utility functions. Moreover, they show that the former approximation result is almost tight, assuming $P \neq NP$.

\subsection*{Our Results and Techniques}

%Equivalence section
In Section~\ref{sec:equivalence},  we examine private signaling in our model in the presence of many states of nature. When the sender's set-function objective is monotone non-decreasing, and the prior distribution is given as input, we show the polynomial-time equivalence of optimal (private) persuasion and the algorithmic problem of maximizing the sender's objective function plus an additive function (subject to no constraints).\footnote{We mention that the reduction from our persuasion problem to the set function maximization problem does not require monotonicity, but the opposite reduction does.} The proof relies on a linear programming formulation of private persuasion, LP duality, and the equivalence of separation and optimization. This leads to polynomial-time persuasion algorithms when the sender's objective is anonymous or supermodular. %{\color{blue} Our results here are simila the line of research on optimal mechanism design \cite{Cai2012a,Cai2012,Cai2013}, in that both relate economic design problems with incentive constraints to purely algorithmic problems. However, our techniques are different from those in \cite{Cai2012a,Cai2012,Cai2013}. }
Our results here are similar to the line of research on optimal mechanism design \cite{Cai2012a,Cai2012,Cai2013}, in that both relate economic design problems with Bayesian incentive constraints to purely algorithmic problems. However, our techniques are different from those in \cite{Cai2012a,Cai2012,Cai2013}.

%Submodular multistate section

Next, we consider a sender with a monotone submodular objective, and efficiently compute a private signaling scheme which is $(1-\frac{1}{e})$-approximately optimal for the sender, modulo an arbitrarily small additive $\epsilon$. This generalizes the result of  \cite{babichenko2016} from two states of nature to many states, modulo the additive loss.  The techniques used in \cite{babichenko2016} do not appear to help in the presence of many states of nature, since they are tied to several structural properties which only hold in the case of a binary state. Therefore, our algorithm uses a different approach, and  crucially relies  on a new structural property of (approximately) optimal private signaling schemes.  Specifically, we prove that there always exists an $\epsilon$-optimal ``simple'' private signaling scheme which is a uniform mixture of polynomially many deterministic schemes, i.e., a scheme that deterministically sends a signal to each receiver upon receiving a state of nature. Notably, this property only depends on the monotonicity of the sender's utility function, and does not rely on submodularity. Using this property, we then use ideas from the literature on submodular function optimization to compute a  $(1-\frac{1}{e})$-approximation to the best ``simple'' scheme.

%Note on sample oracle model
We note that our algorithmic results from Sections~\ref{sec:equivalence} and \ref{sec:submodular} can be approximately extended to the sample oracle model, in which our algorithm is only given sample access to the prior distribution, using  ideas from \cite{Dughmi2016}. The resulting schemes suffer  an arbitrarily small additive loss in the sender's objective and in persuasiveness of their recommendations, and this loss is unavoidable for information theoretic reasons.

%Submodular binary section
In Section \ref{sec:binary}, we examine private signaling in the special case with two states of nature and a monotone submodular objective. We give a simple and explicit construction of a polynomial-time private signaling scheme that serves as a $(1-\frac{1}{e})$-approximation to the optimal scheme, which is the best possible assuming $P \neq NP$ as shown by \cite{babichenko2016}.  This result simplifies, and slightly strengthens, a result of  \cite{babichenko2016}.  Moreover,  the constructed private scheme  has the following distinctive properties: (i) it is \emph{independent} in that it signals independently to each receiver in each of the two states of nature; (ii) it is \emph{oblivious} in the sense that it does not depend on the sender's utility function, and the approximation ratio is guaranteed so long as the function is monotone submodular. To obtain this result, we first use the idea of the \emph{correlation gap} \cite{Agrawal2010} to argue that there always exists an independent  signaling scheme which is a $(1-\frac{1}{e})$ approximation to the optimal private scheme. We then exploit the fact that there are two states of nature to argue that our scheme is the optimal independent scheme, simultaneously for all monotone objectives. We then show two respects in which this result cannot generalize to many states of nature. First, we show that it is NP-hard to obtain a $(1-\frac{1}{e})$-approximation to the best independent signaling scheme in multi-state settings. Second, we show that oblivious schemes cannot guarantee more than a $\frac{1}{\sqrt{m-1}}$ fraction of the optimal sender utility where $m>1$ is the number of states of nature.

%public section
Finally, in Section~\ref{sec:public}, we consider public signaling in our model, and present two negative results. First, we show via a simple example that the optimal private scheme can outperform the optimal public scheme, in terms of maximizing the sender's objective, by a polynomial factor. Second, we employ a reduction from an NP-hard graph coloring promise problem to rule out an algorithm for approximating the optimal public scheme to within any constant factor, and also to rule out an additive PTAS for the problem.

%Unfortunately, for the general case with \emph{polynomially} many states of nature, it becomes NP-hard to approximate the optimal independent signaling scheme within a ratio better than $(1-\frac{1}{e})$. Therefore, the above idea for binary states of nature can only gives a $(1-\frac{1}{e})^2$ approximate signaling scheme -- $(1-\frac{1}{e})$ loss from the ``gap lemma" and $(1-\frac{1}{e})$ loss from computing the optimal independent signaling scheme. 

\section{Preliminaries}
\subsection{Basic Setup}
As in \cite{Arieli2016}, we consider the special case of multi-agent persuasion with binary actions, no inter-agent externalities, and a monotone objective function.  Here, we adopt the perspective of  \emph{sender} facing $n$ \emph{receivers}.  Each receiver has two actions, which we denote by $0$ and $1$. The receiver's payoff depends only on his own action and a random \emph{state of nature} $\theta$ supported on $\Theta$. In particular, we use $u_i(\theta,1)$ and $u_{i}(\theta,0)$ to denote receiver $i$'s utility for action $1$ and action $0$, respectively, at the state of nature $\theta$;  as shorthand, we use $u_i(\theta) = u_i(\theta,1) - u_{i}(\theta,0)$ to denote how much receiver $i$ prefers action $1$ over action $0$ given state of nature $\theta$. Note that $u_i(\theta)$ may be negative.  The sender's utility (our objective) is a function of all the receivers' actions and the state of nature $\theta$.  We use $f_{\theta}(S)$ to denote the sender's utility  when the state of nature is $\theta$ and $S$ is the set of receivers who choose action $1$. We assume throughout the paper that  $f_{\theta}$ is a monotone non-decreasing set function for every $\theta$. For convenience in stating our approximation guarantees, we assume without loss of generality that $f_{\theta}$ is normalized so that $f_\theta(\emptyset)= 0$ and $f_\theta(S) \in [0,1]$ for all $\theta \in \Theta$ and $S \sse [n]$.

As is typical in information structure design, we assume that $\theta$ is drawn from a common prior distribution $\lambda$, 
%\sdmargincomment{I would change the notation to $\lambda(\theta)$ throughout, rather than $\lambda_\theta$} 
that the sender has access to the realized state of nature, and that the sender can publicly \emph{commit} to a policy--- termed a \emph{signaling scheme} --- for  mapping the realized state of nature to a \emph{signal} for each receiver. The signaling scheme may be randomized, and hence reveals noisy partial information regarding the state of nature. The order of events is as follows: (1) The sender commits to a signaling scheme $\varphi$; (2) Nature draws $\theta \sim \lambda$; (3) Signals $ (\sigma_1, \ldots, \sigma_n) \sim \varphi(\theta)$ are drawn, and each receiver $i$ receives the signal $\sigma_i$; (4) Receivers select their actions. 

A general signaling scheme permits sending different signals to different receivers through a private communication channel --- we term these \emph{private signaling schemes} to emphasize this generality. We also study the special case of \emph{public signaling schemes} --- these are restricted to a public communication channel, and hence send the same signal to all receivers. We discuss these two signaling models in Sections \ref{prelim:private} and \ref{prelim:public}, including the equilibrium concept and the induced sender optimization problem for each. In both cases, we are primarily interested in the optimization problem faced by the sender in step (1), the goal of which is to maximize the sender's expected utility. When $\varphi$  yields expected sender utility within an additive [multiplicative] $\eps$ of the best possible, we say it is \emph{$\eps$-optimal} [$\eps$-approximate]  in the additive [multiplicative] sense.

\subsection{Private Signaling Schemes}
\label{prelim:private}
A \emph{private signaling scheme} $\varphi$ is a   randomized map from the set of states of nature $\Theta$ to a set of \emph{signal profiles} $\Sigma= \Sigma_1 \times \Sigma_2...\times \Sigma_n$, where $\Sigma_i$ is the \emph{signal set} of receiver $i$. 
% \sddelete{Note that the commitment is \emph{before} any realization of the nature and the scheme $\varphi$ is publicly known. Upon observing a realized state $\theta$, the sender samples a signal profile $\sigma = (\sigma_1,...,\sigma_n)$ according to $\varphi$ and privately sends signal $\sigma_i \in \Sigma_i$ to receiver $i$.} 
We use $\varphi(\theta,\sigma)$ to denote the probability of selecting the signal profile $\sigma=(\sigma_1,\ldots,\sigma_n) \in \Sigma$ given a state of nature $\theta$. Therefore, $\sum_{\sigma \in \Sigma} \varphi(\theta,\sigma) = 1$ for every $\theta$. With some abuse of notation, we use $\varphi(\theta)$ to denote the random signal profile selected by the scheme $\varphi$ given the state $\theta$. 
%\sddelete{and use $\varphi_{\theta}$ to denote the distribution of the random variable $\varphi(\theta)$}\sdmargincomment{No need to distinguish $\varphi(\theta)$ from $\varphi_\theta$... random variable and distribution are same thing} 
Moreover,  for each $\theta \in \Theta$, $i \in [n]$, and $\sigma_i \in \Sigma_i$, we use $\varphi_i(\theta,\sigma_i) = \Pr[ \varphi_i(\theta) = \sigma_i]$ to denote the marginal probability that receiver $i$ receives signal $\sigma_i$ in state $\theta$.  
%\sddelete{We will use $ \sigma[i]$ to denote the $i$'th component of the signal profile $\sigma$.} \sdmargincomment{I don't like this component business... should be able to do without it}   \sddelete{Note that $ \varphi_i(\theta,\sigma_i) = \sum_{\sigma: \sigma[i] = \sigma_i} \varphi(\theta,\sigma).$} 
An algorithm \emph{implements} a signaling scheme $\varphi$ if it takes as input a state of nature $\theta$, and samples the random variable $\varphi(\theta)$.

Given a signaling scheme $\varphi$, each signal $\sigma_i \in \Sigma_i$ for receiver $i$ is realized with probability $\Pr(\sigma_i) = \sum_{\theta \in \Theta} \lambda(\theta) \varphi_i(\theta,\sigma_i)$. Upon receiving $\sigma_i$, receiver $i$ performs a Bayesian update and infers a posterior belief over the state of nature, as follows: the realized state is $\theta$ with posterior probability $\lambda(\theta) \varphi_i(\theta,\sigma_i)/\Pr(\sigma_i)$. Receiver $i$ then takes the action maximizing his posterior expected utility. In case of indifference, we assume ties are broken in favor of the sender (i.e., in favor of action 1). Therefore, receiver $i$ chooses action $1$ if 
\begin{equation*}
	\frac{1}{\Pr(\sigma_i)}\sum_{\theta \in \Theta} \lambda(\theta) \varphi_i(\theta,\sigma_i) u_i(\theta,1) \geq \frac{1}{\Pr(\sigma_i)} \sum_{\theta \in \Theta} \lambda(\theta) \varphi_i(\theta,\sigma_i) u_i(\theta,0),
\end{equation*}
or equivalently
\begin{equation*}
	\sum_{\theta \in \Theta} \lambda(\theta) \varphi_i(\theta,\sigma_i) u_i(\theta) \geq 0,
\end{equation*}
where $u_i(\theta) = u_i(\theta,1) - u_i(\theta,0)$. %In case of indifference, we assume ties are broken in favor of the sender.

A simple revelation-principle style argument \cite{Kamenica2011,Arieli2016}  shows that there exist an optimal private signaling scheme which is  \emph{direct} and \emph{persuasive}.
%\sdmargincomment{Note the change of terminology --- replace IC with persuasive everywhere.} 
By \emph{direct} we mean that signals correspond to actions --- in our setting $\Sigma_i = \set{0,1}$ for each receiver $i$ --- and can be interpreted as action recommendations. A direct scheme is \emph{persuasive} if the strategy profile where all receivers follow their recommendations forms an equilibrium of the resulting Bayesian game.\footnote{Persuasiveness has also been called \emph{incentive compatibility} or \emph{obedience} in prior work.} Due to the absence of inter-receiver externalities in our setting, such an equilibrium would necessarily also satisfy the stronger property of being a dominant-strategy equilibrium --- i.e., each receiver $i$ maximizes his posterior expected utility by following the recommendation, regardless of whether other receivers follow their recommendations. 
%\sddelete{This is because if two signals result in the same receiver action, we could have named them the same signal without changing the receiver's action. Such a \emph{direct} scheme  has signal profiles $\Sigma = \prod_{i=1}^n \{ 0, 1\}$ where we view $0$ or $1$ as the recommended action.} 

When designing private signaling schemes, we  restrict attention (without loss) to direct and persuasive schemes. Here, a signal profile can be equivalently viewed as a set $S \sse [n]$ of receivers --- namely, the set of receivers who are recommended action $1$. Using this alternative representation, a scheme can be specified by variables $\varphi(\theta,S)$ for all $\theta \in \Theta, S \subseteq [n]$. We can now encode the sender's optimization problem of computing the optimal scheme using the following exponentially-large linear program; note the use of auxiliary variables $x_{\theta,i}$ to denote the marginal probability of recommending action $1$ to receiver $i$ in state $\theta$. 
% Recall that if all the receivers follow the recommendations, the sender's utility after sending signal profile $S$ conditioned on the state $\theta$ is precisely $f_{\theta}(S)$.  
%The marginal probability of recommending action $1$ to receiver $i$ conditioned on the state $\theta$ is $\varphi_i(\theta,1) = \sum_{S:i \in S} \varphi(\theta,S)$, which we denote as $x_{\theta,i}$. Note that $\varphi_i(\theta,0) = 1-x_{\theta,i}$.

%We adopt the perspective of a sender looking to design $\varphi$  to maximize the sender's expected utility. We term the task of computing $\varphi$ the \emph{private Bayesian persuasion}. With variables $\varphi(\theta,S)$ and $x_{\theta,i}$, the optimal direct private signaling scheme can be computed via the following \emph{exponentially-large} linear program. 
%\begin{figure}[H]
%	\centering	
\begin{lp}\label{lp:opt0}
	\maxi{\sum_{\theta \in \Theta} \lambda(\theta) \sum_{S\subseteq [n]} \varphi(\theta,S)  f_{\theta}(S) }
	\st 
	\qcon{\sum_{S:i\in S} \varphi(\theta,S) = x_{\theta,i}}{i \in [n], \theta \in \Theta}
	\qcon{\sum_{\theta \in \Theta}\lambda(\theta) x_{\theta,i} u_i(\theta) \geq 0}{i = 1,...,n}
	\qcon{\sum_{S \subseteq [n]} \varphi(\theta,S) = 1}{\theta \in \Theta}
	\qcon{ \varphi(\theta,S) \geq 0}{\theta \in \Theta; S \subseteq [n]}
\end{lp}
%	Exponentially-Large LP for Computing Optimal Private Signaling Scheme
%\end{figure}

The second set of constraints in LP~\eqref{lp:opt0} are \emph{persuasiveness constraints}, and state that each receiver $i$ should maximize his utility by taking action $1$ whenever that action $1$ is recommended. Note that the persuasiveness constraints for action $0$, which can be written as $\sum_{\theta \in \Theta}\lambda(\theta) (1-x_{\theta,i}) u_i(\theta) \leq 0$ for each $i \in [n]$, are intentionally omitted from this LP.  This omission is without loss when $f_\theta$ is a non-decreasing set function for each $\theta$: any solution to the LP in which a receiver prefers action $1$ when recommended action $0$ can be improved by always recommending action $1$ to that receiver.% recommendingviolating those constraints  can be (weakly) improved by replacing recommendations of action $0$ with recommendations of action $1$ as needed.

\subsection{Public Signaling Schemes}
\label{prelim:public}
A public signaling scheme $\pi$ can be viewed as a special type of private signaling schemes in which each receiver must receive the same signal, i.e., only a public signal is sent. Overloading the notation  of Section~\ref{prelim:private}, we use $\Sigma$ to denote the set of public signals and $\sigma \in \Sigma$ to denote a public signal. A public signaling scheme $\pi$ is fully specified by  $\{ \pi(\theta, \sigma) \}_{\theta, \sigma}$, where  $\pi(\theta, \sigma)$ denotes the probability of sending signal $\sigma$ at state $\theta$.  Upon receiving a signal $\sigma$, each receiver performs the same Bayesian update and infers a posterior belief over the state of nature, as follows: the realized state is $\theta$ with probability $\lambda(\theta) \pi(\theta,\sigma)/\Pr(\sigma)$, where $\Pr(\sigma) = \sum_{\theta \in \Theta}  \pi(\theta,\sigma)$.  This induces a subgame for each signal $\sigma$, one in which all receivers share the same belief regarding the state of nature. 

Whereas in more general settings than ours receivers may play a mixed Nash equilibrium in each subgame, our restriction to a setting with no externalities removes this complication. Given a posterior distribution on states of nature (say, one induced by a signal $\sigma$),  our receivers face disjoint single-agent decision problems, each of which admits an optimal pure strategy. We assume that receivers break ties in favor of the sender (specifically, in favor of action 1), which distinguishes a unique pure response for each receiver. Therefore, our solution concept here distinguishes a unique action profile for each posterior distribution, and hence for each signal. A simple revelation-principle style argument then allows us to conclude that there is an optimal public signaling schemes which is \emph{direct}, meaning that the public signals are action profiles, and \emph{persuasive}, meaning that in the subgame induced by signal $\sigma=(\sigma_1,\ldots,\sigma_n)$ each receiver $i$'s optimal decision problem (which breaks ties in favor of action $1$) solves to action $\sigma_i$. 

Restricting attention to direct and persuasive public signaling schemes, each signal can also be viewed as a subset $S \subseteq [n]$ of receivers taking action $1$. The sender's optimization problem 
%\hxdelete{of computing the optimal public signaling scheme} 
can then be written as the following exponentially-large linear program.
\begin{lp}\label{lp:optPub}
	\maxi{\sum_{\theta \in \Theta} \lambda(\theta) \sum_{S\subseteq [n]} \pi(\theta,S)  f_{\theta}(S) }
	\st 
	\qcon{\sum_{\theta \in \Theta}\lambda(\theta) \pi(\theta,S) \cdot  u_i(\theta) \geq 0}{S \sse [n] \mbox{ with } i \in S}
	\qcon{\sum_{S \subseteq [n]} \pi(\theta,S) = 1}{\theta \in \Theta}
	\qcon{ \pi(\theta,S) \geq 0}{\theta \in \Theta; S \subseteq [n]}
\end{lp}

The first set of constraints are persuasiveness constraints corresponding to action $1$. Note that the persuasiveness constraints for action $0$, which can be written as $\sum_{\theta \in \Theta}\lambda(\theta) \pi(\theta,S) \cdot  u_i(\theta) \leq 0$ for each $S \sse [n]$ and $i \not\in S$, are intentionally omitted from this LP. This omission is without loss when $f_\theta$ is non-decreasing for each state $\theta$: if signal $S$ with $i \not\in S$ is such that receiver $i$ prefers action $1$ in the resulting subgame, then we can replace it with the signal $S \union i$ without degrading the sender's utility.  We remark that LP \eqref{lp:optPub} and LP \eqref{lp:opt0}  only differ in their persuasiveness constraints. %Interestingly, this difference results in sharp contrast between the two  in computational complexity and objective value.  

\subsection{Input Models}

We  distinguish two input models for describing persuasion instances in this paper. The first is the \emph{explicit} model, in which the prior distribution $\lambda$ is given explicitly as a probability vector. The second is the \emph{sample oracle} model, where $\Theta$ and $\lambda$ are provided implicitly through sample access to $\lambda$.   In both models, we assume that given a state of nature $\theta$, we can efficiently evaluate $u_i(\theta)$ for each $i \in [n]$ and $f_\theta(S)$ for each $S \sse [n]$. In the explicit input model, by \emph{computing} a signaling scheme $\varphi$ (whether private or public) we mean that we explicitly list the non-zero variables $\varphi(\theta,S)$ on which the scheme is supported. In the implicit model, computing a signaling scheme $\varphi$ (whether private or public) amounts to providing an algorithm which takes as input a state of nature $\theta$, and samples the random variable~$\varphi(\theta)$.

\subsection{Set Functions and Submodularity}
Given a finite ground set $X$, a \emph{set function} is a map $f: 2^X \to \RR$.  Such a function is  \emph{nonnegative} if $f(S) \geq 0$ for all $S \sse X$,  \emph{monotone non-decreasing} (or \emph{monotone} for short) if $f(S) \leq f(T)$ for all $S \sse T$. Most importantly, $f$ is   \emph{submodular} if  for any $S,T \subseteq X$, we have $f(S \cup T) + f(S \cap T) \leq f(S) + f(T)$. 
%Function $f(S)$ is \emph{monotone} if $f(S) \leq f(T)$ whenever $S \subset T \subseteq X$, and is \emph{non-negative} if $f(S) \geq 0$ for any $S$.

We also consider continuous functions $G$ from the solid hypercube $[0,1]^{X}$ to the real numbers. Such a function is  \emph{nonnegative} if $G(x) \geq 0$ for all $x$, \emph{monotone non-decreasing} (or \emph{monotone} for short) if $G(x) \leq G(y)$ whenever $x \preceq y$ (coordinate wise), and \emph{smooth submodular} (in the sense of \cite{calinescu2011}) if its second partial derivatives exist and are non-positive everywhere.

\vspace{4mm}
\noindent {\bf The Multilinear Extension of a Set Function.} Given any set function $f:2^X \to \RR$, the \emph{multilinear extension} of $f$ is the continuous function $F:[0,1]^X \to \RR$ defined as follows:
\begin{equation}\label{eq:multilinear}
	F(x) = \sum_{S\subseteq X} f(S) \prod_{i \in S}x_i \prod_{i \not \in S}(1-x_i),
\end{equation}
Notice that, $F(x)$ can be viewed as the expectation of $f(S)$ when the random set $S$ independently includes each element $i$ with probability $x_i$. In particular, let $p^I_x$ denote the \emph{independent distribution} with marginals $x$, defined by $p^I_x(S) = \prod_{i \in S}x_i \prod_{i \not \in S}(1-x_i)$, then $F(x) = \Ex_{S \sim p^I_x} f(S)$.
%\sddelete{When $x \in \{0,1\}^{|X|}$, we have $F(x) = f(S_x)$ where $S_{x} = \{ i \in X: x_i =1\}$. } \sdmargincomment{Do we use that (obvious) stuff? It seems not, so I omitted}
If $f$ is nonnegative/monotone then so is $F$. Moreover, if $f$ is submodular then $F$ is smooth submodular. For our results, we will need to maximize $F(x)$ subject to a set of linear constraints on $x$. This problem is NP-hard in general, yet can be approximated by the \emph{continuous greedy process} of  Calinescu et al. \cite{calinescu2011} for fairly general families of constraints. %to approximately maximize $F(x)$  subject to fairly general constraints.
Note that though we cannot exactly evaluate $F(x)$ in polynomial time,  it is sufficient to approximate $F(x)$ within a good precision in order to apply the continuous greedy process. By an additive FPTAS evaluation oracle for $F$, we mean an algorithm that evaluates $F(x)$ within additive error $\eps$ in $\poly(n, \frac{1}{\eps})$ time. %We will call a $\poly(n, \frac{1}{\eps})$ time algorithm an additive FPTAS 

\begin{theorem}[Adapted form \cite{calinescu2011}]\label{thm:greedy}
	%Let $F: [0,1]^n \to [0,1]$ be a non-negative, monotone, smooth submodular function. Let $\P \subseteq [0,1]^n$ be a down-monotone polytope\footnote{A polytope $\P \sse \RRp^n$ is called \emph{down-monotone} if for all $x,y \in \RRp^n$, if $y \in P$ and $x \preceq y$ (coordinate-wise) then $x \in P$.}, specified explicitly by its linear constraints.  Given an evaluation oracle for $F$, there is a polynomial time algorithm that achieves a $(1-\frac{1}{e})$-approximation for the problem of maximizing $F(x)$ subject to $x \in \P$.
	Let $F: [0,1]^n \to [0,1]$ be a non-negative, monotone, smooth submodular function. Let $\P \subseteq [0,1]^n$ be a down-monotone polytope\footnote{A polytope $\P \sse \RRp^n$ is called \emph{down-monotone} if for all $x,y \in \RRp^n$, if $y \in P$ and $x \preceq y$ (coordinate-wise) then $x \in P$.}, specified explicitly by its linear constraints.  Given an additive FPTAS evaluation oracle for $F$,  there is a $\poly(n,\frac{1}{\eps})$ time algorithm that outputs $\bar{x} \in \P$ such that $F(\bar{x}) \geq (1 - \frac{1}{e})OPT - \eps$, where $OPT = \max_{x \in \P} F(x)$. 
\end{theorem}
%\sdcomment{This theorem was vague.. I added the detail. Double check it is correct. Also, we should handle the fact that we can't evaluate $F$ exactly --- how do we do that and how does the error translate? Is an FPTAS for evaluating $F$ enough? Does it have to be additive or multiplicative? One possibility is that we are losing additive loss anyways, so additive loss in evaluating $F$ would translate to additive lost in the approximation ratio and that's OK.}

\vspace{4mm}
\noindent {\bf Correlation Gap.}  A general definition of the correlation gap can be found in \cite{Agrawal2010}. For our results, the following simple definition will suffice. Specifically, for any $x \in [0,1]^{X}$, let $D(x)$ be the set of all distributions $p$ over $2^X$ with fixed marginal probability $\Pr_{S \sim p}(i \in S) = x_i$ for all $i$. Let $p_x^I$, as defined above, be the independent distribution with marginal probabilities $x$. Note that $p_x^I \in D(x)$. For any set function $f(S)$, the correlation gap $\kappa$ is defined as follows:
\begin{equation}\label{eq:corGap}
	\kappa = \max_{x \in [0,1]^X} \max_{p \in D(x)} \frac{ \Ex_{S\sim p} f(S)}{ \Ex_{S \sim p^I_x} f(S)}.
	%\frac{\max_{p \in \D(x)} [  \sum_{S\subseteq X} f(S) \cdot p(S) ] }{\sum_{S\subseteq X} f(S) \cdot \prod_{i \in S}x_i \prod_{i \not \in S}(1-x_i)},
\end{equation}
%where  and  
%In other words, the correlation gap of $f(S)$ is the maximum ratio between the expectation of $f(S)$ when $S$ is drawn from any distribution with a given marginal $x$ and that of $f(S)$ when $S$ includes each element $i$ independently with probability $x_i$. 
Loosely speaking, the correlation gap upper bounds the ``loss" of the expected function value over a random set by ignoring the correlation in the randomness. 

\begin{theorem}[\cite{Agrawal2010}]\label{thm:corgap}
	The correlation gap $ \kappa$ is upper bounded by $\frac{e}{e-1}$ for any non-negative monotone non-decreasing submodular function.
\end{theorem}

\section{Equivalence Between Private Signaling and Objective Maximization}
\label{sec:equivalence}
In this section, we relate the computational complexity of private persuasion to the complexity of maximizing the sender's objective function, and show that the optimal private signaling  scheme can be computed efficiently for a broad class of sender objectives. Let $\F$ denote any collection of monotone set functions. We use $\I (\F)$ to denote the class of all persuasion instances in our model in which the sender utility function $f_\theta$ is in $\F$ for all states of nature $\theta$. We restrict attention to the explicit input model for most of this discussion, though discuss how to extend our results to the sample oracle model, modulo an arbitrarily small additive loss in both the sender's objective and the persuasiveness constraints, at the end of this section.   

The following theorem establishes the polynomial-time equivalence between computing the optimal private signaling scheme and the problem of maximizing the objective function plus an additive function. Note that although the number of variables in LP \eqref{lp:opt0} is exponential in the number of receivers, a vertex optimal solution of this LP is supported on $O(n |\Theta|)$ variables. 

\begin{theorem}\label{thm:equivalence}
	Let $\F$ be any collection of monotone set functions. There is a polynomial-time algorithm which computes the optimal private signaling scheme given any  instance in $\I (\F)$ if and only if there is a polynomial time algorithm for maximizing $f(S) + \sum_{i\in S} w_i$ given any $f \in \F$ and any set of weights $w_i \in \RR$. 
\end{theorem} 
\begin{proof}
	We first reduce optimal private signaling to maximizing the objective function plus an additive function, via linear programming duality.  In particular,  consider the following dual program of LP \eqref{lp:opt0}  with variables $w_{\theta,i}, \alpha_i, y_{\theta}$.
	\begin{lp}\label{lp:opt0Dual}
		\mini{\sum_{\theta \in \Theta} y_{\theta}}
		\st 
		\qcon{\sum_{i \in S}w_{\theta,i} + y_{\theta} \geq \lambda(\theta) f_{\theta}(S)}{S \subseteq [n], \theta \in \Theta}
		\qcon{w_{\theta,i} + \alpha_i \lambda(\theta) u_i(\theta) = 0} {i = 1,...,n}
		\qcon{\alpha_i \geq 0}{ i \in [n]}
	\end{lp}
	We can obtain a  separation oracle for   LP \eqref{lp:opt0Dual} given an algorithm  for maximizing $f_{\theta}(S)$ plus an additive function. Given any variables $w_{\theta,i}, \alpha_i, y_{\theta}$, separation over the first set of constraints reduces to maximizing the set function $g_\theta(S) = f_{\theta}(S) - \frac{1}{\lambda(\theta)} \sum_{i\in S} w_{\theta,i}$ for each $\theta \in \Theta$. The other constraints can be checked directly in linear time.  Given the resulting separation oracle, we can use the Ellipsoid method to obtain a vertex optimal solution to both  LP \eqref{lp:opt0Dual} and its dual LP \eqref{lp:opt0} in polynomial time \cite{GLSbook}.
	%\begin{equation}\label{equ:SetMax}
	%\frac{y_{\theta}}{\lambda(\theta)} \geq f_{\theta}(S) - \frac{1}{\lambda(\theta)} \sum_{i\in S} w_{\theta,i} \, \qquad  \forall S \subseteq [n], \theta \in \Theta.
	%\end{equation}
	% To check Inequalities in \eqref{equ:SetMax} , we maximize $\big[ f_{\theta}(S) - \frac{1}{\lambda(\theta)} \sum_{i\in S} w_{\theta,i} \big]$ over $S \subseteq [n]$, and check whether the optimal objective value is at most $y_{\theta}/\lambda(\theta)$  for every $\theta$. If so, then all these constraints are satisfied. Otherwise, we know that  $\max_{S \subseteq [n]}\big[ f_{\theta}(S) - \frac{1}{\lambda(\theta)} \sum_{i\in S} w_{\theta,i} \big] > y_{\theta}/\lambda(\theta)$ for some $\theta$, therefore  the optimal solution $S^*$ to the maximization problem, together with the $\theta$, corresponds to a violated constraint. To sum up, we can design a polynomial-time separation oracle for LP \eqref{lp:opt0Dual} given access to a polynomial-time algorithm to maximize any $f_{\theta} \in \F$ plus an arbitrary additive function. By the equivalence between separation and optimization, we can also solve LP \eqref{lp:opt0Dual}, thus its dual LP \eqref{lp:opt0}, in polynomial time. Note that the returned optimal solution will have polynomial-sized support.  
	
	%\sdcomment{Shaddin redid this proof. Haifeng, please check everything below carefully.}

	We now prove the converse.  Namely, we construct a polynomial-time Turing reduction from the problem of maximizing $f$ plus an additive function to a private signaling problem in $\I (\F)$. At a high level, we first reduce the set function maximization problem to a certain linear program, and then prove that solving the dual of  the LP reduces to optimal private signaling for a set of particularly constructed instances in $\I (\F)$. 
	
	Given $f\in F$ and weights $w$, our reduction concerns the following linear program, parameterized by $\mathbf{a} = (a_1,...,a_n)$ and $b$, with variables $\mathbf{z} = (z_1,...,z_n)$ and $v$.   
	\begin{lp}\label{lp:ClosureDual}
		\mini{\sum_{ i \in  [n]} a_i z_i + b v} 
		\st
		\qcon{\sum_{i \in S} z_i +  v \geq f(S) }{ S  \subseteq  [n]}
	\end{lp}
	Let  $\P$ denote the feasible region of LP \eqref{lp:ClosureDual}. As the first step of our reduction,  we reduce maximizing the set function $g_w(S) = f(S) + \sum_{i\in S} w_i$   to the separation problem for $\P$. Let $z_i = -w_i$ for each $i$.  Notice that $(\mathbf{z}, v) $ is feasible (i.e., in $\P$) if and only if $v \geq \max_{S \subseteq [n]} f(S) -  \sum_{i \in S} z_i$.   Therefore, we can binary search for a value  $\tilde{v}$ such that $(\mathbf{z},\tilde{v})$ is almost feasible, but not quite. More precisely,  let $B$ denote the bit complexity of the $f(S)$'s and the $w_i$'s,  then binary search returns the exact optimal value of the set function maximization problem after $\O(B)$ steps. We then set $\tilde{v}$ to equal that value minus $2^{-B}$. Feeding $(\mathbf{z}, \tilde{v})$ to the separation oracle, we obtain a violated constraint which must correspond to the maximizer of  $f(S) + \sum_{i \in S} w_i$.
	
	As the second step of our reduction, we reduce the separation problem for $\P$ to solving LP \eqref{lp:ClosureDual} for every choice of objective coefficients $\mathbf{a}$ and $b$. This polynomial-time Turing reduction follows from the equivalence of separation and optimization \cite{GLSbook}.
	
	Third, we reduce solving LP \eqref{lp:ClosureDual} for arbitrary $\mathbf{a}$ and $b$ to the special case where $\mathbf{a} \in [0,1]^n$ and $b=1$. The reduction involves a case analysis. (a) If any of the objective coefficients are negative, then the fact that $\P$ is upwards closed implies that LP \eqref{lp:ClosureDual} is unbounded. (b) If $b=0$ and $a_i >0$ for some $i$, then the LP is unbounded since we can make $v$ arbitrarily small and $z_i$ arbitrarily large. Normalizing by dividing by $b$, we have reduced the problem to the case when $b=1$ and $a \succeq 0$ (coordinate-wise). (c) Now suppose that $a_i>1=b$ for some $i$; the LP is unbounded by making $z_i$ arbitrarily small and $v$ arbitrarily large. This analysis leaves the case of $b=1$ and $\mathbf{a} \in [0,1]^n$.
	
	Fourth, we reduce   LP \eqref{lp:ClosureDual} with parameters $\mathbf{a} \in [0,1]^n$ and $b=1$ to its dual shown below, with variables $p_S$ for $S \sse [n]$.
	\begin{lp}\label{lp:Closure}
		\maxi{\sum_{S \subseteq [n]} p_S f(S)} 
		\st
		\qcon{\sum_{S: i \in S} p_{S} \leq a_i }{ i \in [n]}
		\con{\sum_{S \subseteq [n]} p_S  = 1}
		\qcon{p_S \geq 0}{S \subseteq [n]}
	\end{lp}
	We note that LP~\eqref{lp:Closure} is not the standard dual of LP~\eqref{lp:ClosureDual}; in particular the first set of constraints are inequality rather than equality constraints. It is easy to see that LP~\eqref{lp:Closure} is equivalent to the standard dual  when $f$ is monotone non-decreasing, and  that an optimal solution to one of the two duals can be easily converted to an optimal solution of the other.

	The fifth and final step of our reduction will reduce  LP \eqref{lp:Closure} to a private signaling problem in $\I(\F)$.  There are $n$ receivers and two states of nature $\theta_0, \theta_1$ with $\lambda(\theta_0) = \lambda(\theta_1) = 1/2$. Define $u_i(\theta_1) = 1$ and $u_i (\theta_0) = -\frac{1}{a_i}$ ($-\infty$ if $a_i = 0$) for all $i$. The sender's utility function satisfies $f_{\theta_1} = f_{\theta_0} = f$. Let $\varphi^*$ be an optimal signaling scheme, in particular an optimal solution to the instantiation of LP \eqref{lp:opt0} for our instance. Note that all receivers prefer action 1 in state $\theta_1$; therefore, it is easy to see that $\varphi^*$ can be weakly improved, without violating the persuasiveness constraints, by modifying it to always recommend action 1 to all receivers when in state $\theta_1$. After this modification, $\varphi^*$ is an optimal solution to the following LP, which optimizes over all signaling schemes satisfying $\varphi(\theta_1,[n])=1$.
	\begin{lp}\label{lp:optInstance}
		\maxi{\frac{1}{2}f([n]) + \frac{1}{2}  \sum_{S\subseteq [n]} \varphi(\theta_0,S)  f(S) }
		\st 
		\qcon{\sum_{S:i\in S} \varphi(\theta_0,S) = x_{\theta_0,i}}{i \in [n]}
		\qcon{  x_{\theta_0,i}  \leq a_i}{i = 1,...,n}
		\con{\sum_{S \subseteq [n]} \varphi(\theta_0,S) = 1}
		\qcon{ \varphi(\theta_0,S) \geq 0}{\theta \in \Theta; S \subseteq [n]}
	\end{lp}  
	It is now easy to see that setting $p_S = \varphi^*(\theta_0, S)$ yields an optimal solution to LP~\eqref{lp:Closure}
\end{proof}

As an immediate corollary of Theorem \ref{thm:equivalence}, the optimal private signaling scheme can be computed efficiently  when the sender's objective function is supermodular or anonymous. Recall that a set function $f:2^{[n]} \to \RR$ is anonymous if there exists a function $g: \ZZ \to \RR$ such that $f(S) = g(|S|)$. 
%\hxdelete{For example, the supermajority function used for describing the outcome of majority voting is anonymous. }
\begin{corollary}\label{cor:equiv}
	There is a polynomial time algorithm for computing the optimal private signaling scheme when the sender objective functions are either supermodular or anonymous.
\end{corollary}
\begin{proof}
	Since a supermodular function plus an additive function is still supermodular, and the problem of unconstrained supermodular maximization can be solved in polynomial time, Theorem \ref{thm:equivalence} implies that the optimal private signaling scheme can also be computed in polynomial time. As for anonymous objectives, there is a simple algorithm for maximizing an anonymous set function plus an additive function. In particular, consider the problem of maximizing $f(S) + \sum_{i\in S} w_i$ where $f(S) = g(|S|)$. Observe that fixing $|S| = k$, the optimal set $S_k$ corresponds to the $k$ highest-weight elements in $w$.  Enumerating all $k$ and choosing the best $S_k$ yields the optimal set.%We can then $i \in [n]$ and choose the one that induces the optimal objective value.
\end{proof}

Finally, we make two remarks on Theorem \ref{thm:equivalence}, particularly on the reduction from optimal private signaling to set function maximization. First, we note that the assumption of monotonicity is not necessary to the reduction from signaling to optimization. In other words, even without the monotonicity assumption for  the sender's objective function, one can still efficiently compute the optimal private signaling scheme for instances in $\I(\F)$ given access to an  oracle for maximizing $f(S) + \sum_{i \in S} w_i$ for any  $f \in \F$ and weight vector $w$.  This can be verified by adding back the persuasiveness constraints for action $0$ to LP \eqref{lp:opt0} and examining the corresponding dual, which has similar structure to LP \eqref{lp:opt0Dual}. We omit the trivial details here.   Consequently, Corollary \ref{cor:equiv} applies to non-monotone supermodular or anonymous functions as well. 

Second, our reduction assumes that the prior distribution $\lambda$ over the state of nature is explicitly given. This can be generalized to the sample oracle model. In particular, when our only access to $\lambda$ is through random sampling, we can implement an $\eps$-optimal and $\eps$-persuasive\footnote{This is the natural relaxation of persuasiveness. In our setting, a receiver approximately maximizes his posterior expected payoff, to within an additive $\epsilon$, by following the scheme's recommendation. This is regardless of the actions of other receivers.}  private  signaling scheme in $poly(n,\frac{1}{\eps})$ time using an idea from \cite{Dughmi2016} (assuming $u_{i}(\theta) \in [-1,1]$). The algorithm is as follows: given any input state $\theta$, we first take $poly(n,\frac{1}{\eps})$ samples from $\lambda$, and then solve LP \eqref{lp:opt0}  on the empirical distribution of the samples plus $\theta$, with relaxed (by $\epsilon$) persuasiveness  constraints. Finally, we signal for $\theta$ as the solution to the LP suggests. The analysis of this algorithm is very similar to that in \cite{Dughmi2016}, thus is omitted here. Moreover,  the bi-criteria loss is inevitable in this oracle model due to information theoretic reasons \cite{Dughmi2016}.  
%Moreover, the $\epsilon$ loss in the incentive compatibility is inevitable due to information theoretic reasons. In particular, for any $c < 1$ and integer $K$ there exists a private persuasion instance with $2$ states and $1$ receiver  such that any exactly incentive compatible and $c$-optimal signaling scheme must take more than $K$ samples (see the constructed instance for Theorem 5.5 (a) in \cite{Dughmi2016}).    

\section{Private Signaling with Submodular Objectives}\label{sec:submodular}
In this section we consider optimal private signaling for submodular sender objectives, and show that there is a polynomial time $(1-\frac{1}{e})$-approximation scheme, modulo an additive loss of $\eps$. This is almost the best possible: \citet{babichenko2016} show that even in the special case of two states of nature, it is NP-hard to approximate the optimal private signaling scheme within a factor better than $(1-\frac{1}{e})$ for monotone submodular sender objectives. %    as shown by as shown by This is almost tight, up to the \sdedit{additive loss of $\epsilon$}, due to a hardness result by  \citet{babichenko2016}: \sdedit{they show that even in the special case of} two states of nature, it is NP-hard to approximate the optimal private signaling scheme within a factor better than $(1-\frac{1}{e})$ for monotone submodular sender objectives.   

\begin{theorem}\label{thm:multi}
	Consider private signaling with monotone submodular sender objectives. Let $OPT$ denote the optimal sender utility. For any $\epsilon > 0$, a private signaling scheme achieving expected sender utility at least $(1-\frac{1}{e})OPT - \epsilon$ can be implemented in $\poly(n,|\Theta|, \frac{1}{\epsilon})$ time. %In particular, when $u = \poly(\epsilon)$,  a signaling scheme achieving sender utility $(1-\frac{1}{e})(OPT - \epsilon)$ can be implemented in $\poly(n,\frac{1}{\epsilon})$ time.
\end{theorem}

The main technical challenge in proving Theorem \ref{thm:multi}  is that a private signaling scheme may have exponentially large support, as apparent from linear program \eqref{lp:opt0}. To overcome this difficulty, we prove 
%Our algorithm for Theorem \ref{thm:multi} crucially relies on 
a structural characterization of (approximately) optimal persuasive private schemes, i.e., solutions to LP \eqref{lp:opt0}.  Roughly speaking, we show that LP \eqref{lp:opt0} always has an approximately optimal solution with polynomial-sized support and nicely structured distributions.  This greatly narrows down the solution space we need to search over.    
%Before formally stating the characterization, we define the notion of K-uniformness.    
% $x_{\theta,i}$'s are fully determined by $\varphi(\theta,S)$'s in LP \eqref{lp:opt0}, we therefore use $\{\varphi(\theta,S) \}_{\theta \in \Theta, S \in 2^{[n]}}$, or simply $\varphi$ for convenience, to denote a feasible solution to LP \eqref{lp:opt0}.  
Recall that for any $\theta$, $\varphi(\theta)$ is a random variable supported on $2^{[n]}$. We say $\varphi(\theta)$ is  \emph{$K$-uniform} if it follows a uniform distribution on a multiset of size $K$. The following lemma exhibits a structural property regarding (approximately) optimal solutions to LP \eqref{lp:opt0}. Notably, this property only depends on monotonicity of the sender's objective functions and does not depend on submodularity. Its proof is postponed to the end of this section. 

\begin{lemma}\label{lem:existence}
	Let $f_{\theta}$  be monotone for each $\theta$. For any $\epsilon >0$, there exists an $\epsilon$-optimal persuasive private signaling scheme $\bar{\varphi}$ such that  $\bar{\varphi}(\theta)$ is $K$-uniform for every $\theta$, where $K = \frac{108n\log (2n|\Theta|)}{\epsilon^3}$.
\end{lemma}

By Lemma \ref{lem:existence}, we can, without much loss,  restrict our design of $\varphi(\theta)$ to the special class of $K$-uniform distributions. Note that a  $K$-uniform distribution $\varphi(\theta)$ can be described by variables $x_{\theta,i}^j \in \{0,1\}$ for $i \in [n], j\in [K]$, where $x_{\theta,i}^j$ denotes the recommended action to receiver $i$ in the $j$'th profile in the support of $\varphi(\theta)$. Relaxing our variables to lie in $[0,1]$, this leads to optimization problem \eqref{prog:multirelaxed}, where $F_{\theta} (x) = \sum_{S\subseteq [n]} f_{\theta}(S) \prod_{i \in S} x_{i} \prod_{i \not \in S} (1-x_{i})$  is the multi-linear extension of $f_{\theta}$.

%This results in the following optimization program \eqref{prog:multirelaxed}, in which $x_{\theta,i}^j$'s are variables  and $F_{\theta} (x) = \sum_{S\subseteq [n]} f_{\theta}(S) \prod_{i \in S} x_{i} \prod_{i \not \in S} (1-x_{i})$ for any $x \in [0,1]^n$ is the multi-linear extension of $f_{\theta}$.  We remark that if we restrict  $x_{\theta,i}^j \in \{0,1\}$ in optimization program \eqref{prog:multirelaxed}, it becomes the problem of computing optimal private schemes restricted to $K$-uniform $\varphi(\theta)$'s; \hxedit{for each $j=1,...,K$, $\{ x_{\theta,i}^j \}_{i=1}^n \in \{0,1 \}^n$ constitutes a profile in the support of $\varphi(\theta)$.} This will be made formal in the following Claim \ref{claim:approxOPT} and its proof.

\begin{lp}\label{prog:multirelaxed}
	\maxi{\sum_{\theta \in \Theta} \frac{\lambda(\theta)}{K} \sum_{j=1}^K F_{\theta}(x_{\theta}^{j})}
	\st 
	\qcon{\sum_{\theta \in \Theta} \frac{\lambda(\theta)}{K} \sum_{j=1}^K  x_{\theta,i}^j u_i(\theta) \geq 0}{i = 1,...,n}
	\qcon{0 \leq x_{\theta,i}^j \leq 1}{i=1,...,n;\theta \in \Theta}
\end{lp} 

As a high level, our algorithm  first approximately solves Program \eqref{prog:multirelaxed} and then signals according to its solution. Details are in Algorithm \ref{alg:multi}, which we instantiate with $\epsilon > 0$ and  $K = \frac{108n\log (2n|\Theta|)}{\epsilon^3}$.  Since $F_{\theta}(x) = \Ex_{S \sim p_x^I} f(S)$ where $p^I_x$ is the independent distribution over $2^{[n]}$ with marginal probability $x$, the expected sender utility induced by the signaling scheme in Algorithm \ref{alg:multi} is  precisely the objective value of Program \eqref{prog:multirelaxed} at the obtained solution. Theorem \ref{thm:multi} then follows from two claims:  1. The optimal objective value of Program \eqref{prog:multirelaxed} is $\epsilon$-close to the optimal sender utility (Claim \ref{claim:approxOPT} ); 2. The continuous greedy process \cite{calinescu2011} can be applied to Program \eqref{prog:multirelaxed} to efficiently compute a $(1-1/e)$-approximate solution, modulo a small additive loss (Claim \ref{claim:multi}). We remark that  Theorem \ref{thm:multi} can be generalized to the sample oracle model, but with an additional $\eps$-loss in persuasiveness constraints (assuming $u_{i}(\theta) \in [-1,1]$), using the idea from \cite{Dughmi2016}.

\begin{algorithm}
	\begin{algorithmic}[1]
		\PARAMETER  $\eps > 0$
		\INPUT Prior distribution $\lambda$ supported  on $\Theta$
		\INPUT $u_i(\theta)$'s and value oracle access to the sender utility $f_{\theta}(S)$
		\INPUT State of nature $\theta$ %$= (\bvec{s}(\theta), \bvec{r}(\theta))$, with $\bvec{s}(\theta), \bvec{r}(\theta) \in [-1,1]^{n}$.
		\OUTPUT A set $S \subseteq [n]$ indicating the set of receivers who will be recommended action $1$.% (other receivers will be recommended action $0$).
		
		\STATE Approximately solve Program \eqref{prog:multirelaxed}. 
		%Employ the continuous greedy algorithm to approximately solve the following optimization program with $K = \frac{18(\log n+1)}{\epsilon^2 u}$. 
		Let $\{ \tilde{x}_{\theta,i}^j \}_{\theta \in \Theta, i \in [n], j\in[K]}$ be the returned solution.  
		
		\STATE Choose $j$ from $[K]$ uniformly at random; For each receiver $i$, add $i$ to $S$ independently with probability $\tilde{x}_{\theta,i}^j$.
		\STATE Return $S$.
	\end{algorithmic}
	\caption{Private Signaling Scheme for Submodular Sender Objectives}
	\label{alg:multi}
\end{algorithm}

\begin{claim}\label{claim:approxOPT}
	When $K = \frac{108n\log (2n|\Theta|)}{\epsilon^3}$, the optimal objective value of Program \eqref{prog:multirelaxed} is at least $OPT - \epsilon$, where $OPT$ is the optimal sender utility in private signaling. 
\end{claim}
\begin{proof}
	By Lemma \ref{lem:existence}, there exists a private signaling scheme  $\bar{\varphi}$ such that: (i) $\bar{\varphi}$ achieves sender utility at least $OPT -\epsilon$; (ii) for each $\theta$, there exists $K$ sets $S_{\theta}^1,...,S_{\theta}^K \subseteq [n]$ such that $\bar{\varphi}_{\theta}$ is a uniform distribution over $\{S_{\theta}^1,...,S_{\theta}^K \}$. 
	%More precisely, $\bar{\varphi}(\theta,S_{\theta}^j) = \frac{1}{K}$ for all $j \in [K], \theta \in \Theta$, and $\bar{\varphi}(\theta,S) = 0$ otherwise. 
	Utilizing $\bar{\varphi}$,  we can construct a feasible solution $\bar{x}$ to Program \eqref{prog:multirelaxed} with objective value at least $OPT- \eps$.  In particular, let $\bar{x}_{\theta}^j \in \{0,1 \}^n$ be the indicator vector of the set $S_{\theta}^j$, formally defined as follows: $\bar{x}_{\theta,i}^j = 1$ if and only if $i \in S_{\theta}^j$. By referring to the feasibility of $\bar{\varphi}$ to LP \eqref{lp:opt0},  it is easy to check that $\bar{x}_{\theta,i}^j$'s are feasible to Program \eqref{prog:multirelaxed}. Moreover, since $F_{\theta}(\bar{x}_{\theta}^j) = f_{\theta}(S_{\theta}^j)$,  the objective value of Program \eqref{prog:multirelaxed} at the solution $\bar{x}$ equals the objective value of Program \eqref{lp:opt0} at the solution $\bar{\varphi}$, which is at least $OPT - \epsilon$. Therefore, the optimal objective value of Program  \eqref{prog:multirelaxed} is at least $OPT-\epsilon$, as desired.
\end{proof}

\begin{claim}\label{claim:multi}
	There is an algorithm that runs in $\poly(n,|\Theta|,K,\frac{1}{\eps})$ time  and computes a $(1-1/e)$-approximate solution, modulo an  additive loss of $\epsilon/e$,  to Program \eqref{prog:multirelaxed}. 
\end{claim}
\begin{proof}
	The objective function of Program \eqref{prog:multirelaxed} is a linear combination, with non-negative coefficients,  of multilinear extensions of monotone submodular functions, thus is smooth, monotone and submodular. Moreover, the function value can be evaluated within error $\eps$ by $\poly(n,\frac{1}{\eps})$ random samples, thus in $\poly(n,\frac{1}{\eps})$  time. To apply Theorem \ref{thm:greedy}, we only need to prove that the feasible region is a down-monotone polytope. Observe that there always exists an optimal solution to Program \eqref{prog:multirelaxed} such that $x_{\theta,i} = 1$ for any $\theta,i$ such that $u_i(\theta) \geq 0$. Therefore, w.l.o.g., we can pre-set these variables to be $1$ and view the program as an optimization problem over $x_{\theta,i}$'s for all $\theta,i$ such that $u_i(\theta)<0$. It is easy to check that these $x_{\theta,i}$'s form a down-monotone polytope determined by polynomially many constraints, as desired. 
\end{proof}

%\begin{remark}
%When the prior distribution $\lambda$ over $\Theta$ is given via a sampling oracle, Theorem \ref{thm:multi} still holds, but with an additional $\eps$-loss in persuasiveness constraints (assuming $u_{i}(\theta) \in [-1,1]$). The algorithm is to, given any input state $\theta$, take $poly(n,\frac{1}{\eps})$ samples from $\lambda$ and then run Algorithm \ref{alg:multi}  using the empirical distribution of the samples plus $\theta$, with relaxed (by $\epsilon$) persuasiveness  constraints in Program \eqref{prog:multirelaxed}.  The analysis  is similar to that in \cite{Dughmi2016}.  
%\end{remark}

\subsection*{Proof of  Lemma \ref{lem:existence}}
Our proof is based on the probabilistic method. Recall that the optimal private signaling scheme can be computed by solving  the exponentially-large LP \eqref{lp:opt0}.  Roughly speaking, given any optimal private scheme $\varphi^*$, we will take polynomially many samples from $\varphi^*(\theta)$ for each $\theta$, and prove that with strictly positive probability the corresponding empirical distributions form a solution to LP \eqref{lp:opt0} that is close to optimality. However, the sampling approach usually suffers from $\epsilon$-loss in both the objective and persuasiveness constraints. It turns out that  the $\epsilon$-loss to persuasiveness constraints can be avoided in our setting with carefully designed pre-processing steps. 

At a high level, to get rid of the $\epsilon$-loss in persuasiveness constraints, there are two main technical barriers. The first is to handle the estimation error in the receiver's utilities, which is inevitable due to sampling. We address this by adjusting the $\varphi^*$ to strengthen the persuasiveness constraints so that a small estimation error would still preserve the original persuasiveness constraints. The second barrier arises when some $x^*_{\theta,i}$'s are smaller than inverse polynomial of the precision $\epsilon$, then $poly(\frac{1}{\epsilon})$ samples cannot guarantee a good multiplicative estimate of  $x^*_{\theta,i}$. We deal with this issue by making ``honest" recommendation, i.e.,  action $0$, at these cases, and show that such modification will not cause much loss to our objective.

We first introduce some convenient notations. For any receiver $i$, let set $\Theta^+_i = \{ \theta: u_{i}(\theta) \geq 0 \}$ be the set of states at which receiver $i$ (weakly) prefers action $1$; Similarly, $\Theta^{-}_i = \{ \theta: u_i(\theta) < 0 \}$ be the set of states at which receiver $i$ prefers action $0$.   Moreover, for any state of nature $\theta$, let $I^+_{\theta} = \{i: u_{i}(\theta) \geq 0 \}$ be the set of receivers who (weakly) prefer action $1$ at state $\theta$.  It would be convenient to think of $\{\Theta^+_i \}_{i \in [n]}$ and $\{ I^+_{\theta}\}_{\theta \in \Theta}$  as two different partitions of the set $\{(\theta,i): u_{i}(\theta) \geq 0 \}$. %Similarly, define  $u_i^- = \sum_{\theta \in \Theta^-_i } \lambda(\theta)u_{i}(\theta)$. 

Observe that by monotonicity there always exists an optimal signaling scheme $\varphi^*$ such that $x^*_{\theta,i} = 1$ for every $\theta \in \Theta_i^+$. Let $\varphi^*$ be such an optimal signaling scheme and $OPT$ denote the optimal sender utility. We now adjust the scheme $\varphi^*$ such that they do not degrade the objective value by much but is more suitable for applying concentration bounds for our probabilistic argument.

\vspace{3mm}
\noindent { \bf Adjustment 1: Always Recommend Action $0$ When $x^*_{\theta,i} < \frac{\epsilon}{3n}$}
\vspace{1mm}

Note that $x^*_{\theta,i} < \frac{\epsilon}{3n}$ only when $\theta \in \Theta_i^-$, i.e., action $0$ is the best action for receiver $i$ conditioned on $\theta$. We first adjust $\varphi^*$ to obtain a new scheme $\tilde{\varphi}$, as follows: $\tilde{\varphi}$ is the same as $\varphi^*$ except that for every $\theta,i$ such that $x^*_{\theta,i} < \frac{\epsilon}{3n}$, $\tilde{\varphi}$ always recommends action $0$ to receiver $i$ given the state of nature $\theta$. As a result, $\tilde{x}_{\theta,i}$ equals  $x^*_{\theta,i}$ whenever  $x^*_{\theta,i} \geq \frac{\epsilon}{3n}$ and equals $0$ otherwise. Note that the signaling scheme still satisfies the persuasiveness constraints. 

Naturally, each adjustment above, corresponding to $\theta,i$ satisfying $x^*_{\theta,i} < \frac{\epsilon}{3n}$, could decrease the objective value since the marginal probability of recommending action $1$ decreases. Nevertheless, this loss, denoted as $L(\theta,i)$, can be properly bounded as follows:  
\begin{eqnarray*}
	L(\theta,i) 
	&=& \lambda(\theta) \cdot \bigg[ \sum_{S: i \in S} \varphi^*(\theta,S) f_{\theta}(S)- \sum_{S: i \in S}  \varphi^*(\theta,S) f_{\theta}(S\setminus\{i\}) \bigg] \\
	& \leq & \lambda(\theta) \cdot \bigg[ \sum_{S: i \in S} \varphi^*(\theta,S) \bigg] \\
	& = & \lambda(\theta)  x_{\theta,i}^* \leq \frac{\lambda(\theta) \epsilon }{3n}.\\
	%		& \leq & \frac{\lambda(\theta) \epsilon }{3n}.
\end{eqnarray*}

As a result, the aggregated loss of all the adjustments made in this step can be upper bounded by $\sum_{\theta \in \Theta}\sum_{i=1}^{n} \frac{\lambda(\theta) \epsilon }{3n} = \frac{\eps}{3}$. That is, the objective value of $\tilde{\varphi}$ is at least $OPT - \frac{\eps}{3}$. 

\vspace{3mm}
\noindent { \bf Adjustment 2: Strengthen the Persuasive Constraints by Scaling Down $x_{\theta,i}$'s} 
\vspace{1mm}

We now strengthen the persuasiveness constraints by further adjusting the $\tilde{\varphi}$ obtained from above so that a small estimation error due to sampling will still maintain the original persuasiveness constraints. For any $\theta$, we define $\varphi' (\theta,S) = \frac{3}{3+\epsilon}\tilde{\varphi}(\theta,S) $ for all $S \not = I^+_{\theta} $, and define $\varphi'(\theta,I^+_{\theta})= 1-\sum_{S \not = I^+_{\theta}} \varphi'(\theta,S)$. Obviously, $\varphi'_{\theta}$ is still a distribution over $2^{[n]}$. Moreover, we claim that $x'_{\theta,i}=\Ex_{S \sim \varphi'_{\theta}} \mathbb{I}(i \in S)  = 1$ whenever $\tilde{x}_{\theta,i} = 1$, i.e., $  \theta \in \Theta_i^+$. That is, given state $\theta$, any receiver $i \in I_{\theta}^+$ will still aways be recommended action $1$.  This is because, to construct $\varphi'_{\theta}$, we moved some probability mass from all other sets $S$ to the set $I_{\theta}^+$, therefore the marginal probability of recommending action $1$ to any receiver $i \in I_{\theta}^+$ will not decrease. However, this marginal probability is originally $1$ in the solution of $\tilde{\varphi}$. Therefore, $x'_{\theta,i} $ still equals $1$ for any $i \in I_{\theta}^+$, or equivalently, for any $\theta \in \Theta_i^{+}$.  Similarly, we also have  $x'_{\theta,i}=0$ whenever $\tilde{x}_{\theta,i} = 0$.

Let $Val(\varphi)$ denote the objective value of a  scheme $\varphi$.  We claim that $Val(\varphi') \geq OPT - \frac{2\eps}{3}$ and $\varphi'$ satisfies $x'_{\theta,i} = \frac{3}{3+\epsilon} \tilde{x}_{\theta,i}$ for every $\theta \in \Theta_i^-$.   For any $i \in [n], \theta \in \Theta_i^-$ (which means $i \not \in I_{\theta}^+$), we have
\begin{equation*}
	x'_{\theta,i} = \sum_{S: i \in S} \varphi'(\theta,S) = \frac{3}{3+\epsilon} \sum_{S: i \in S} \tilde{\varphi} (\theta,S) = \frac{3}{3+\epsilon} \tilde{x}_{\theta,i},
\end{equation*}
since the summation excludes the term $\varphi'(\theta,I_{\theta}^+)$. We now prove the guarantee of the objective value. Observe that $\varphi'(\theta,I^+_{\theta}) \geq \frac{3}{3+\epsilon} \tilde{\varphi}(\theta,I^+_{\theta})$ also holds in our construction.  Therefore, we have 
\begin{eqnarray*}
	Val(\varphi') & = & \sum_{\theta \in \Theta} \lambda(\theta) \sum_{S\subseteq [n]} \varphi'(\theta,S)  f_{\theta}(S) \\
	&\geq & \frac{3}{3+\epsilon} \sum_{\theta \in \Theta} \lambda(\theta) \sum_{S\subseteq [n]} \tilde{\varphi}(\theta,S)  f_{\theta}(S) \\
	&= & \frac{3}{3+\epsilon} \cdot Val(\tilde{\varphi}) \\
	%			  &=&  Val(\tilde{\varphi}) - \frac{\epsilon}{3+\epsilon}Val(\tilde{\varphi})\\
	%			  & \geq  & Val(\tilde{\varphi}) -  \frac{\epsilon}{3} \\
	& \geq & OPT - \frac{2\eps}{3},
\end{eqnarray*}
where we used the  upper bound $Val(\tilde{\varphi}) \leq 1$.

\vspace{3mm}
\noindent   {\bf Existence of An $\epsilon$-Optimal Solution of Small Support.} 
\vspace{1mm}

The above two steps of adjustment result in a feasible $\frac{2\eps}{3}$-optimal solution $\varphi'$ to LP \eqref{lp:opt0} that satisfies the following properties: (i) $x'_{\theta,i} = x^*_{\theta,i} = 1$ whenever $u_i(\theta) \geq 0$; (ii) $x'_{\theta,i} = \frac{3}{3+\epsilon} \tilde{x}_{\theta,i} = \frac{3}{3+\epsilon} x^*_{\theta,i} \geq \frac{\eps}{4n}$ when $x^*_{\theta,i} \geq \frac{\eps}{3n}$ and $\theta \in \Theta_i^-$; (iii) $x'_{\theta,i} =0$ when $x^*_{\theta,i} < \frac{\eps}{3n}$ and $\theta \in \Theta_i^-$.    Utilizing such a $\varphi'$ we show that there exists an $\eps$-optimal solution $\bar{\varphi}$ to LP \eqref{lp:opt0} such that the distribution $\bar{\varphi}_{\theta}$ is a $K$-uniform distribution for every $\theta$, where  $K = \frac{108n\log (2n|\Theta|)}{\epsilon^3}$. 

%\begin{proof}
Our proof  is based on the probabilistic method. For each $\theta$, independently take $K = \frac{108n\log (2n|\Theta|)}{\epsilon^3}$ samples from random variable $\varphi'(\theta)$, and let $\bar{\varphi}_{\theta}$ denote the corresponding empirical distribution. Obviously, $\bar{\varphi}_{\theta}$ is a $K$-uniform distribution. We claim that with strictly positive probability over the randomness of the samples,  $\bar{\varphi}$ is feasible to LP \eqref{lp:opt0} and achieves utility at least $Val(\varphi')-\frac{\epsilon}{3} \geq OPT - \eps$.

We first examine the objective value.  Observe that the objective value $Val(\varphi')$ can be viewed as the expectation of the random variable $\sum_{\theta \in \Theta} \lambda(\theta)f_{\theta}(S_{\theta}) \in [0,1]$,  where $S_{\theta}$ follows the distribution of  $\varphi'(\theta)$. Our sampling procedure generates  $K$ samples for the random variable $\{S_{\theta}\}_{\theta \in \Theta}$, therefore by the Hoeffding bound, with probability at least $1 -  \exp(-2K\epsilon^2/9) > 1- 1/(2n|\Theta|) $, the empirical mean is at least $Val(\varphi')-\epsilon/3$. 

Now we only need to show that all the persuasiveness constraints are preserved with high probability. First, observe that if $x'_{\theta,i} = 0$, then $\bar{x}_{\theta,i}$ induced by $\bar{\varphi}$ also equals  $0$. This is because $x'_{\theta,i} = \Ex_{S \sim \varphi'(\theta)}\mathbb{I}(i \in S) = 0$ implies that $i$ is not contained in any $S$ from the support of $\varphi'(\theta)$, therefore, also not contained in any sample. Similarly,  $x'_{\theta,i} = 1$ implies $\bar{x}_{\theta,i} = 1$. To show that all the persuasiveness constraints hold, we only need to argue that $\bar{x}_{\theta,i} \leq x^*_{\theta,i}$ for every $\theta \in \Theta_i^-$ satisfying $x^*_{\theta,i} \geq \frac{\eps}{3n}$. This holds with high probability by tail bounds. In particular,  $x'_{\theta,i} = \Ex_{S \sim \varphi'(\theta)}\mathbb{I}(i \in S) $ and we take $K$ samples from $\varphi'(\theta)$. By the Chernoff bound, with probability at least 
$$1- \exp(-\frac{K\epsilon^2 x'_{\theta,i}}{27}) \geq 1- \exp(-\frac{K\epsilon^3}{108n}) > 1- \frac{1}{2n|\Theta|},$$ the empirical mean $\bar{x}_{\theta,i}$ is at most $(1+\epsilon/3) x'_{\theta,i} = x^*_{\theta,i}$.

Note that there are at most $n|\Theta|$ choices of such $\theta,i$. By union bound, with probability  at least $1- (n|\Theta|+1)/(2n|\Theta|) > 0$, $\bar{\varphi}$  satisfies all the persuasiveness constraints thus is feasible to LP \eqref{lp:opt0}, and achieves objective value at least $Val(\varphi') - \frac{\epsilon}{3} \geq OPT - \eps$. So there must exist a feasible $\eps$-optimal solution $\bar{\varphi}$ to LP \eqref{lp:opt0} such that $\bar{\varphi}_{\theta}$ is $K$-uniform for every $\theta$.  This concludes our proof of Lemma \ref{lem:existence}.

\section{An Oblivious Private Scheme for Binary-State Settings}\label{sec:binary}
In this section, we consider the special case with two states of nature, denoted by $\theta_0$ and $\theta_1$, and submodular sender utility functions. 
%Arieli and Babichenko \cite{Arieli2016} examine this setting and exhibit several properties regarding the optimal signaling schemes. Babichenko and Barman \cite{babichenko2016} propose an algorithm that computes an $(1-1/e-\eps)$-approximate signaling scheme and runs in $\poly(n,\frac{1}{\eps})$ time. 
We show that a $(1-\frac{1}{e})$-approximate private signaling scheme can be  \emph{explicitly} constructed. The approximation ratio is tight by \cite{babichenko2016}. Moreover, the constructed signaling scheme has the following distinctive properties:  (i) it signals independently
to each receiver, which we term an \emph{independent signaling scheme};
(ii) it is \emph{oblivious} in the sense that it does not depend on the sender's utility function so long as it is monotone submodular.  %All proofs in this section are omitted due to space constraints, and can be found in the full version on arXiv.

In particular, we consider the following \emph{independent signaling scheme} $\varphi_{I}$.   Given any state of nature $\theta \in \{\theta_0, \theta_1 \}$,  $\varphi_{I}$  recommends action $1$ to receiver $i$ independently with probability $x_{\theta,i}$ for each $i \in [n]$ and recommends action $0$ otherwise, where $x_{\theta,i}$'s are defined as follows:
\begin{equation} \label{eq:binaryScheme}
\begin{array}{ll}
\mbox{If $u_i(\theta_0) < 0$ and $u_i(\theta_1) < 0$,} & \mbox{then $x_{\theta,i} =0$ for any $\theta \in \{ \theta_0,\theta_1\}$} \\
\mbox{If $u_i(\theta_0) \geq 0$ and $u_i(\theta_1) \geq 0$,} & \mbox{then $x_{\theta,i} = 1$ for any $\theta \in \{ \theta_0,\theta_1\}$.} \\
\mbox{If $u_i(\theta_0) \geq 0$ and $u_i(\theta_1) < 0$,} & \mbox{ then $x_{\theta_0,i} =1$ and $x_{\theta_1,i} = \min \{ -\frac{ \lambda(\theta_0) u_i(\theta_0)}{\lambda(\theta_1) u_i(\theta_1)}, 1 \}$.} \\	
\mbox{If $u_i(\theta_0) < 0$ and $u_i(\theta_1) \geq 0$,} & \mbox{ then $x_{\theta_1,i} =1$ and $x_{\theta_0,i} = \min \{ -\frac{ \lambda(\theta_1) u_i(\theta_1)}{\lambda(\theta_0) u_i(\theta_0)}, 1 \}$.}
\end{array}
\end{equation}
 
Observe that $\varphi_{I}$  is precisely the scheme that maximizes the probability of recommending action $1$ to each receiver $i$ independently. Notice that $\varphi_I$  is only determined by the receiver's payoffs, and does \emph{not} depend on the sender's payoff function $f_{\theta}(S)$. Nevertheless, the following theorem shows that $\varphi_I$ is approximately optimal.

\begin{theorem}\label{thm:binary}
	In the binary-state setting with monotone submodular sender objectives, the scheme $\varphi_I$ is a $(1-\frac{1}{e})$-approximate private signaling scheme.    
\end{theorem}
%Theorem \ref{thm:binary} follows from the following two lemmas. The proof of Lemma \ref{lem:relax} employs the idea of the correlation gap from the robust stochastic optimization literature \cite{Agrawal2010}.

To prove Theorem \ref{thm:binary}, we first exhibit an optimization program, i.e., Program \eqref{prog:relaxed}, that computes the optimal independent signaling scheme, i.e., the independent scheme that achieves the maximum sender utility among all independent signaling schemes.  Here the variables $x_{\theta,i}$ denote the probability of  recommending action $1$ to bidder $i$  given the state of nature $\theta$. Notice that $\sum_{S\subseteq [n]} f_{\theta}(S) \prod_{i \in S} x_{\theta,i} \prod_{i \not \in S} (1-x_{\theta,i})$ in the objective is precisely the multi-linear extension of $f_{\theta}$ at  $x_{\theta} = (x_{\theta,1},...,x_{\theta,n})^T$. Theorem \ref{thm:binary} follows from two lemmas: 
 1. The optimal independent signaling  scheme  serves as a $(1-\frac{1}{e})$-approximation to the optimal private signaling scheme (Lemma \ref{lem:relax}); 2. When there are two states of nature, the $\varphi_I$ defined above is an optimal independent signaling scheme (Lemma \ref{lem:OptSol2}).
\begin{figure}[H]
	\centering
		\begin{lp}\label{prog:relaxed}
			\maxi{\sum_{\theta \in \Theta} \lambda(\theta) \sum_{S\subseteq [n]} f_{\theta}(S) \prod_{i \in S} x_{\theta,i} \prod_{i \not \in S} (1-x_{\theta,i})}
			\st 
			\qcon{\sum_{\theta \in \Theta}\lambda(\theta) x_{\theta,i} u_i(\theta) \geq 0}{i = 1,...,n}
			\qcon{0 \leq x_{\theta,i} \leq 1}{i=1,...,n;\theta \in \Theta}
		\end{lp}
Optimal Independent Signaling Problem
\end{figure}

\begin{lemma}\label{lem:relax}
	In the multiple-state setting with monotone submodular sender objectives, the optimal independent signaling scheme is a $(1-\frac{1}{e})$-approximation to the optimal private signaling scheme.  
	
	% Let $\{ \tilde{x}_{\theta,i} \}_{\theta\in \Theta,i\in[n]} $ be the optimal solution to Program \eqref{prog:relaxed}. Then the following independent signaling scheme serves as a $(1-1/e)$-approximation to the optimal private signaling scheme: given any realized state $\theta$,  \emph{independently} recommend action $1$ to receiver $i$ with probability $\tilde{x}_{\theta,i}$ and action $0$ otherwise.
\end{lemma}
\begin{proof}
	Let $\varphi^*$ be any optimal private signaling scheme and $OPT$ denote the optimal sender utility. 
	%Recall that $\varphi^*(\theta)$ denotes a random set $S$ and $\varphi^*_{\theta}$ is the distribution of $\varphi^*(\theta)$. 
	Therefore,  the principle's optimal utility can be expressed as follows.  
	\begin{equation*}
	OPT = \sum_{\theta \in \Theta} \lambda(\theta) \sum_{S \subseteq [n]} f_{\theta}(S) \varphi^*(\theta, S)  = \sum_{\theta \in \Theta} \lambda(\theta) \Ex_{S \sim \varphi^*(\theta) } f_{\theta}(S).
	\end{equation*}
	
	Let $x^*_{\theta,i} = \sum_{i \in S} \varphi^*(\theta,S)$ denote the marginal probability of recommending action $1$ to receiver $i$ by $\varphi^*$ given the state of nature $\theta$. Since $\varphi^*$ is a direct signaling scheme and satisfies the persuasiveness constraints, so $\{ x^*_{\theta,i} \}_{\theta \in \Theta, i\in[n]}$ is a feasible solution to Program \eqref{prog:relaxed}.  Let $Val( x^*)$ denote the corresponding objective value of $\{ x^*_{\theta,i} \}_{\theta \in \Theta, i\in[n]}$ in Program \eqref{prog:relaxed}. We shall prove that $Val(x^*) \geq (1-\frac{1}{e})OPT$.  
	
	Let $x_{\theta}$ denote the vector $ (x_{\theta,1},...,x_{\theta,n})^T$. Observe that the objective function of Program \eqref{prog:relaxed} is precisely $\sum_{\theta \in \Theta} \lambda(\theta) \Ex_{S \sim p_{x_{\theta}}^I} f_{\theta}(S) $ where $p_{x_{\theta}}^I$ is the independent distribution with marginal $x_{\theta}$. 
	%satisfying $p_{x_{\theta}}^I(S) = \prod_{i\in S} x_{\theta,i} \prod_{i \not \in S}(1-x_{\theta,i})$. 
	Utilizing the correlation gap for monotone submodular functions, we have
	\begin{equation*}
	\frac{ \Ex_{S \sim \varphi^*_{\theta} } f_{\theta}(S)  }{ \Ex_{S \sim p_{x_{\theta}^*}^I  } f_{\theta}(S) } \leq \max_{p \in D_{x_{\theta}^*}}  \frac{ \Ex_{S \sim p } f_{\theta}(S)  }{ \Ex_{S \sim p_{x_{\theta}^*}^I  } f_{\theta}(S) } \leq \frac{e}{e-1}
	\end{equation*}
	That is, $\Ex_{S \sim p_{x_{\theta}^*}^I } f_{\theta}(S) \geq (1-\frac{1}{e}) \Ex_{S \sim  \varphi^*_{\theta}   } f_{\theta}(S)$ for every $\theta$. 
	Therefore, 
	\begin{equation*}
	\frac{Val(x^*)}{OPT} = \frac{ \sum_{\theta \in \Theta} \lambda(\theta) \cdot \Ex_{S \sim p_{x_{\theta}^*}^I } f_{\theta}(S)  }{ \sum_{\theta \in \Theta} \lambda(\theta) \cdot \Ex_{S \sim \varphi^*_{\theta} } f_{\theta}(S) } \geq 1-\frac{1}{e}.
	\end{equation*} 
	
	Note that the optimal objective value of Program \eqref{prog:relaxed} is at least $Val(x^*)$ thus is at least  $ (1-1/e)OPT$. Therefore if $\{ x_{\theta,i} \}_{\theta \in \Theta,i\in[n]}$ is the optimal solution to Program \eqref{prog:relaxed}, then the scheme, which independently recommends action $1$ to receiver $i$ with probability $x_{\theta,i}$ and action $0$ otherwise at state $\theta$,  achieves sender utility at least $(1-\frac{1}{e})OPT$, therefore is an $(1-\frac{1}{e})$-approximation to the optimal private signaling scheme.

	\end{proof}

%\subsection{Warm-up: Binary State of Nature}
%We first consider the case where there are only two states of nature, or equivalently, $\theta$ is binary. This case is considered by Babichenko and Barman \cite{babichenko2016} who proposed a $(1-1/e-\eps)$-approximation for this setting. Here we present a new $(1-1/e)$-approximation algorithm, simply by combining Lemma \ref{lem:relax} with the following observation. Notice that the ratio $(1-1/e)$ is tight, due to an approximation hardness by Babichenko and Barman \cite{babichenko2016}. 
% Our algorithm then follows from Lemma \ref{lem:relax} and the following observation.
\begin{lemma}\label{lem:OptSol2}
	In the  binary-state  setting with monotone sender objectives, the $\varphi_I$ defined above is an optimal independent signaling scheme. 
\end{lemma}
	\begin{proof}
 First observe that the objective function of Program \eqref{prog:relaxed} is non-decreasing in $x_{\theta,i}$ due to the monotonicity of $f_{\theta}(S)$ and the non-negativity of $\lambda(\theta)$. The optimality of $\varphi_I$ simply follows from the monotonicity.
In particular,  when $u_i(\theta_0) < 0$ and $u_i(\theta_1) < 0$, any feasible solution to Program \eqref{prog:relaxed} must satisfy $x_{\theta,i} =0$ for any $\theta$, i.e., receiver $i$ can never be persuaded to take action $1$. On the other hand, when $u_i(\theta_0) \geq 0$ and $u_i(\theta_1) \geq 0$, the sender can always recommend action $1$ to receiver $i$, so we have $x_{\theta,i} = 1$ for any $\theta$. Now consider the case $u_i(\theta_0) \geq 0$ and $u_i(\theta_1) < 0$. If $x_{\theta_0,i} < 1$ at optimality, then resetting its value to $1$ neither decreases the objective value due to monotonicity nor violates the persuasiveness constraint since $u_i(\theta_0) \geq 0$. Moreover, given that $x_{\theta_0,i} = 1$, the maximum possible value of $x_{\theta_1,i}$ that still maintains feasibility is  precisely $\min \{ -\frac{ \lambda(\theta_0) u_i(\theta_0)}{\lambda(\theta_1) u_i(\theta_1)}, 1 \}$, as we defined in \eqref{eq:binaryScheme}. The case for $u_i(\theta_0) < 0$ and $u_i(\theta_1) \geq 0$ is similarly derived. Thus $\varphi_I$ is optimal to Program \eqref{prog:relaxed} .
\end{proof}

We conclude this section with two negative results regarding generalizing Theorem \ref{thm:binary} to many states of natures. We first show that it is NP-hard to  approximate the optimal independent signaling scheme within a factor better than $(1-\frac{1}{e})$ 
when there are multiple states of nature (Proposition \ref{prop:multiHard}). We then prove that the best oblivious scheme can perform poorly in terms of the sender utility when there are multiple states  (Proposition \ref{prop:multiObliv}).

\begin{proposition}\label{prop:multiHard} 
	In the multiple-state setting with monotone submodular sender objectives, it is NP-hard to compute an independent signaling scheme that $(1-1/e)$-approximates the optimal independent scheme.
\end{proposition}
\begin{proof}
	We reduce the Submodular Welfare Problem (SWP) to solving Program \eqref{prog:relaxed}. The SWP is a classic problem of combinatorial auctions, described as follows: given $n$ \emph{items} and $m$ \emph{bidders} with monotone submodular utility functions $f_i : 2^{[n]} \to \RR^+$, we seek a partition of the items into disjoint sets $S_1,...,S_m$, each for a bidder, in order to maximize the total welfare $\sum_{i=1}^m f_i(S_i)$. This problem is first studied by  \citet{Khot2005}. They prove that it is NP-hard to approximate the optimal welfare within a factor better than  $(1-\frac{1}{e})$. To prove the proposition, we provide a reduction from SWP to  solving Program \eqref{prog:relaxed}.
	
	Consider an SWP instance with $n$ items and $m$ bidders. Bidder $j$'s utility function  $f_j : 2^{[n]} \to \RR^+$  is monotone submodular. We now construct a private persuasion instance. For the reduction, it helps to think of  items as receivers and bidders as states of nature. In particular, consider a persuasion instance with $n$ receivers and $m+1$ equally possible states of nature $\theta_0,...,\theta_m$, where $\theta_j$ corresponds to bidder $j$ for $j >0$. For each $\theta_j$ with $j > 0$, define $u_i(\theta_j) = -1$ for all receiver $i$ and the sender's utility function conditioned on state $\theta_j$ as $f_{j}$. State $\theta_0$ is a special state with $u_i(\theta_0) = 1$ for all receive $i$ and $f_{\theta_0}(S) = 0$ for all set $S$. We claim that the optimal objective value of Program \eqref{prog:relaxed} for our constructed private persuasion instance equals, up to a scale factor $1/(m+1)$, the optimal welfare of the given SWP instance. 
	
	By monotonicity, w.l.o.g., we can assume $x_{\theta_0,i} = 1$ for all $i$ at optimality since $u_i(\theta_0) >0$. After  pre-setting the values of these variables, Program \eqref{prog:relaxed} can be re-written as follows. 
	\begin{lp*}
		\maxi{\frac{1}{m+1}\sum_{j=1}^m  F_{j}(x_{\theta_j}) }
		\st 
		\qcon{\sum_{j=1}^n x_{\theta_j,i}  \leq  1}{i = 1,...,n}
		\qcon{0 \leq x_{\theta_j,i} \leq 1}{i=1,...,n; \, j = 1,...,m}
	\end{lp*}
	
	By viewing $x_{\theta_j,i}$ as the probability of allocating item $i$ to player $j$, the above optimization program corresponds precisely to the given SWP instance above, up to a scaling of the objective, except that the allocation is now allowed to be randomized. In particular, the constraints form a partition matroid polytope. Since any fractional allocation can be efficiently rounded to a deterministic allocation achieving the same welfare (e.g., use Pipage rounding \cite{calinescu2011}), an $\alpha$-approximate solution to Program \eqref{prog:relaxed} yields an $\alpha$-approximate deterministic allocation for the SWP instance. Since it is NP-hard to approximate SWP within a ratio better than $(1-\frac{1}{e})$, the same approximation hardness holds for solving Program \eqref{prog:relaxed}, completing the proof. 
	\end{proof}

%We defer the proof of Proposition \ref{prop:multiHard} to the Appendix. Our next result shows that when there are multiple states of nature,  the best oblivious scheme can achieve at most $\frac{1}{\sqrt{m-1}}$ fraction of  the optimal sender utility in private signaling schemes, where $m = |\Theta|$. We leave open the intriguing question of whether there is a constant approximation oblivious private scheme when the number of states of nature is a also constant. 

\begin{proposition}\label{prop:multiObliv}
	For any integer $m>1$, there exists an instance with $m$ states of nature such that any oblivious private scheme can achieve at most $\frac{1}{\floor{\sqrt{m-1}}}$ fraction of the optimal sender utility.
\end{proposition}
\begin{proof}
		At a high level, we will construct a class of persuasion problems with the same receiver payoffs but different sender utility functions. Therefore any oblivious private signaling scheme should be the same across all the constructed instances. However we prove, via Yao's principle, that it will perform poorly in at least one of these instances. 
		
		Consider the  persuasion problem with $n$ receivers and $m+1 = n^2 + 1$ states of nature, denoted as $\theta_0, \theta_1,...,\theta_{m}$, each occurring with equal probability. The sender's utility function for each state is one of the following two simple functions: $f^0 \equiv 0$ and $f^1$ satisfying $f^1(S)=0$ when $S = \emptyset$ and $f^1(S) = 1$ otherwise. We reserve the flexibility of choosing the sender's objective function from one of these two functions at each state. The state $\theta_0$ is a special state with $u_i(\theta_0) = 1$ for every $i$ and sender utility function $f_{\theta_0} = f^0$ always. Let  $u_i(\theta_j) = -1$ for every $j >0$  and every $i \in [n]$, therefore the receivers' payoffs are fixed. We now define the possible different ways of defining the sender's objective function at different states. In particular,  let \emph{assignment} $A: \{\theta_1,...,\theta_m \} \to \{f_0,f_1\}$ denote a deterministic mapping that assigns an objective function to each state $\theta_j$ for $j > 0$ and satisfies that  exactly $n$ states are mapped to function $f_1$. We use $\A$ to denote all such assignments. Since each $A$ also uniquely determines a concrete persuasion instance, we will also think of $\A$ as the set of all the  constructed instances.  
		
		For any instance $A \in \A$ and a private signaling scheme $\varphi$ for $A$, we use $u(\varphi,A)$ to denote the sender's expected utility by using  scheme $\varphi$ for instance $A$. Observe that, for any instance $A$, the optimal private signaling scheme is to recommend action $1$ to every receiver at state $\theta_0$, recommend action $1$ to exactly one  receiver, distinctly across different states, at each state with sender objective function $f^1$, and recommend action $0$ otherwise. The optimal sender utility is $\frac{n}{m+1}$. On the other hand, since all the instances in $\A$ have the same receiver payoffs, these instances will use  the same oblivious private signaling scheme. Let $\P$ denote the set of all private signaling schemes $\varphi = \{\varphi(\theta,S) \}_{\theta,S}$ for any instance in $\A$. Note that $\P$ only depends on receivers' payoffs, thus is the same for any instance in $\A$. 
		%In fact, $\P$ is essentially the feasible region of LP \eqref{lp:opt0}. 
		For any $\varphi \in \P$, the approximation ratio of $\varphi$ as an oblivious scheme is $\min_{A \in \A} \frac{u(\varphi,A)}{n/(m+1)}$.    By Yao's principle, we have $$ \min_{A \in \A} u(\varphi,A) \leq \max_{\varphi \in \P} \mathbf{E}_{A \sim \D} [u(\varphi,A)], \qquad  \forall  \varphi \in \P$$ where $\D$ is an arbitrary distribution over set $\A$.  
		%Here we used the fact that $\P$ is convex and the sender's utility $u(\varphi,A)$ is linear in $\varphi$. 
		Now, let $\D$ be the uniform distribution over $\A$, then the best oblivious signaling scheme is simply to persuade $n$ receivers separately at arbitrarily chosen $n$ states since states are symmetric and indistinguishable under distribution $\D$. Some calculation shows that the expected utility is $\frac{1}{m+1}=\max_{\varphi \in \P} \mathbf{E}_{A \sim \D} [u(\varphi,A)]$ for uniform distribution $\D$.   Therefore, the best approximation ratio of any oblivious scheme is $$\max_{\varphi \in \P} \min_{A \in \A} \frac{u(\varphi,A)}{n/(m+1)} \leq \frac{m+1}{n} \cdot \frac{1}{m+1} = \frac{1}{n}.$$
\end{proof}

\section{Inefficacy and Hardness of Public Persuasion}
\label{sec:public}
In this section, we turn our attention to the design of  optimal public signaling schemes, and show a stark contrast with private signaling, both in terms of their efficacy at optimizing the sender's utility, and in terms of their computational complexity. 

We start with an example illustrating how the restriction to public signaling can drastically reduce the sender's expected utility. The example is notably simple: two states of nature, and a binary sender utility function which is independent of the state of nature. We show a multiplicative gap of $\Omega(n)$, and an additive gap of $1-\frac{1}{\Omega(n)}$,  between the expected sender utility from the optimal private and public signaling schemes,  where $n$ is the number of receivers.

\begin{example}[Inefficacy of Public Signaling Schemes]\label{ex:inefficacy}
	Consider an instance with $n$ identical receivers and two states of nature $\Theta = \{ \mathbf{H},\mathbf{L} \}$. Each receiver has the same utility function, defined as follows: $u_i(\mathbf{H}) = 1$ and $u_i(\mathbf{L}) = -1$, for all $i$. The state of nature $\mathbf{H}$ occurs with probability $\frac{1}{n+1}$, and $\mathbf{L}$ occurs with probability $\frac{n}{n+1}$.  The sender's utility function is $f_{\theta}(S) = f(S) = 
	\min(|S|,1)$. In other words, the sender gets utility $1$ precisely when at least one receiver takes action $1$.
	
	The persuasiveness constraints imply that each receiver can take action $1$ with probability no more than $\frac{2}{n+1}$. This is achievable by always recommending action $1$ to the receiver in state $\mathbf{H}$, and recommending action $0$ with probability $\frac{1}{n}$ in state $\mathbf{L}$. The sender's expected utility depends on how these recommendations are correlated.

	The optimal private scheme anti-correlates the receivers' recommendations in order to guarantee that at least one receiver takes action $1$ always, which achieves an expected sender utility of $1$, the maximum possible. Specifically, in state $\mathbf{H}$ the scheme always recommends action $1$ to every receiver, and in state $\mathbf{L}$ the scheme chooses one receiver uniformly at random and recommends action $1$ to that receiver, and action $0$ to the other receivers.
	
	We argue that no public scheme can achieve sender utility more than $\frac{2}{n+1}$. Indeed, since receivers are identical, our solution concept implies that they choose the same action for every realization of a public signal. Therefore, the best that a  public scheme can do is to recommend action $1$ to all receivers simultaneously with probability $\frac{2}{n+1}$ in aggregate, and recommend action $0$ with the remaining probability, yielding an expected sender utility of $\frac{2}{n+1}$.  This is achievable:  in state $\mathbf{H}$ the scheme always recommends action $1$ to every receiver, and in state $\mathbf{L}$ the scheme recommends action  $1$ to all receivers with probability $\frac{1}{n}$, and action $0$ to all receivers with probability $1-\frac{1}{n}$.

\end{example}

Our next result illustrates the computational barrier to obtaining the optimal public signaling scheme, even for additive sender utility functions. Our proof is inspired by a reduction in \cite{mixture_selection} for proving the hardness of computing the best posterior distribution over $\Theta$, a problem termed \emph{mixture selection} in \cite{mixture_selection}, in a voting setting. That reduction is from the maximum independent set problem. Since a public signaling scheme is a combination of posterior distributions, one for each signal, we require a more involved reduction from a graph-coloring problem to prove our result.

\begin{theorem}\label{thm:PublicHard}
	Consider public signaling in our model, with sender utility function $f_{\theta}(S) = f(S) =\frac{|S|}{n}$. It is NP-hard to  approximate the optimal  sender utility to within any constant multiplicative factor. Moreover, there is no additive PTAS for evaluating the optimal  sender utility, unless P = NP.
	%a constant $\delta (\geq \frac{2}{9})$ such that it is NP-hard to approximate the optimal sender utility \emph{additively} within $\delta$. 
\end{theorem}
\begin{proof}
	We prove by reducing from the following NP-hard problem. In particular, \cite{Khot2012} prove that for any positive integer $k$, any integer $q$ such that $q \geq  2^k+ 1$, and an arbitrarily small constant $\eps > 0$,  given an undirected graph $G$, it is NP-hard to distinguish between the following two cases:
	\begin{itemize}
		\item {\bf Case 1}: There is a $q$-colorable induced subgraph of $G$ containing a  $(1-\eps)$ fraction of all vertices, where each color class contains a $ \frac{1-\eps}{q}$ fraction of all vertices.
		\item{\bf Case 2}: Every independent set in $G$ contains less than a $\frac{1}{q^{k+1}}$ fraction of all vertices.
	\end{itemize} 
	Given graph $G$ with vertices $[n] = \set{1,\ldots,n}$ and edges $E$, we will construct a public persuasion instance so that the desired algorithm for approximating the optimal sender utility can be used to distinguish these two cases. Our construction is similar to that  in \cite{mixture_selection}. We let there be $n$ receivers, and let $\Theta= [n]$. In other words, both receivers and states of nature correspond to vertices of the graph. We fix the uniform prior distribution over states of nature --- i.e., the realized state of nature is a uniformly-drawn vertex in the graph. We define the receiver utilities as follows: $u_{i}(\theta) = \frac{1}{2}$ if $i = \theta$;  $u_{i}(\theta) =-1$ if $(i,\theta) \in E$;  and $u_{i}(\theta) = -\frac{1}{4n}$ otherwise. We define the sender's utility function, with range $[0,1]$, to be $f_{\theta}(S) = f(S) =\frac{|S|}{n}$. The following claim is proven in \cite{mixture_selection}. 
	\begin{claim}[\cite{mixture_selection}]\label{claim:SignalIndep}
		For any distribution $x \in \Delta_{\Theta}$, the set $S = \{ i \in [n]: \sum_{\theta \in \Theta} x_{\theta} u_i(\theta) \geq 0 \}$ is an independent set of $G$. 
	\end{claim}

	% As a high level, given any graph, we will construct an instance so that any desired approximation algorithm for evaluating the optimal sender utility can be used to distinguish between these two cases. We use a similar construction as used by a reduction in \cite{mixture_selection}.  In particular, given any graph $G=(V,E)$ with $n$ vertices, construct an instance with $n$ receivers and $\Theta = [n]$. The state of nature follows a uniform distribution over $\Theta$. In other words, each receiver corresponds to a node and the state of nature $\theta$ is also a node chosen uniformly at random. Define $u_{i}(\theta)$ as follows: $u_{i}(\theta) = \frac{1}{2}$ if $i = \theta$;  $u_{i}(\theta) =-1$ if $(i,\theta) \in E$;  and $u_{i}(\theta) = -\frac{1}{4n}$ otherwise. The sender's utility function $f_{\theta}(S) = \frac{|S|}{n} \in [0,1]$ for any $\theta$. The following claim is proved in \cite{mixture_selection}. 
	% \begin{claim}[\cite{mixture_selection}]\label{claim:SignalIndep}
	% For any $x \in \Delta_{|\Theta|}$, the set $S = \{ i \in [n]: \sum_{\theta \in \Theta} x_{\theta} u_i(\theta) \geq 0 \}$ is an independent set of $G$. 
	% \end{claim}

	Claim \ref{claim:SignalIndep} implies that upon receiving any public signal with any posterior distribution $x$ over $\Theta$, the players who take action $1$ always form an independent set of $G$. Therefore, if the graph $G$ is from {\bf Case 2}, the sender's expected utility in any public signaling scheme is at most $\frac{1}{q^{k+1}}$.
	
	%We use $I_{\max}$ to denote the size of the largest independent set in $G$. Claim \ref{claim:SignalIndep} implies that upon receiving any public signal with any posterior distribution $x$ over $\Theta$, the players who take action $1$ always form an independent set. Therefore, if the graph $G$ is from {\bf Case 2}, the sender's optimal utility in any public scheme is \emph{at most} $\frac{1}{q^{k+1}}$ since the utility conditioned on any signal is at most $I_{\max}/n \leq \frac{1}{q^{k+1}}$ by definition.  

	Now supposing that $G$ is from {\bf Case 1}, we fix the corresponding coloring of $(1-\epsilon)n$ vertices with colors $k=1,\ldots,q$, and we use this coloring to construct a public scheme achieving expected sender utility at least $\frac{(1-\epsilon)^2}{q}$. The scheme uses $q+1$ signals, and is as follows: if $\theta$ has color $k$ then deterministically send the signal $k$, and if $\theta$ is uncolored then deterministically send the signal $0$. Given signal $k > 0$, the posterior distribution on states of nature is the uniform distribution over the vertices with color $k$ --- an independent set $S_k$ of size $\frac{1-\epsilon}{q} n$. It is easy to verify that receivers $ i \in S_k$ prefer action $1$ to action $0$, since $\sum_{\theta \in S_k} \frac{1}{|S_k|} u_i(\theta) = \frac{1}{|S_k|} ( \frac{1}{2} - \frac{|S_k| - 1}{4n}) > \frac{1}{4 |S_k|} \geq 0$. Therefore, the sender's utility is  $f(S_k) = \frac{|S_k|}{n} = \frac{1-\epsilon}{q}$ whenever $k>0$. Since signal $0$ has probability  $\epsilon$, we conclude that the sender's expected utility is at least $\frac{(1-\epsilon)^2}{q}$, as needed.

	% We now prove that if the graph is from {\bf Case 1},  the optimal send utility  is \emph{at least} $(1-\eps)^2/q$. If particular, if $G$ is from  {\bf Case 1}, consider the following public signaling scheme with $q+1$ signals: send signal $i$ deterministically whenever $\theta$ corresponds to a node with color $i$, for $i = 1,...,q$; send signal $0$ otherwise. Upon receiving a public signal $i >0$, the posterior distribution is a uniform distribution over all nodes with color $i$. It is easy to verify that in this case any receiver corresponding to a node with color $i$ is persuaded to take action $1$ and all other receivers are persuaded to take action $0$. Therefore, the sender's utility conditional on any signal $i >0$ is $(1-\eps)\frac{1}{q}$. Each such signal occurs with probability also $(1-\eps)\frac{1}{q}$. Therefore, the sender's expected utility is at least $q \times [(1-\eps)\frac{1}{q}]^2 = (1-\eps)^2/q$, as desired. 
	
	Since distinguishing {\bf Case 1}  and {\bf Case 2} is NP-hard for arbitrarily large constants  $k$ and $q$, we conclude that it is NP-hard to approximate the optimal sender utility to within any constant factor. Moreover, by setting $k=1,q=3$, we conclude that the sender's utility cannot be approximated additively to within $(1-\eps)^2/3 - 1/3^2 > 1/9$, and thus there is no additive PTAS, unless P=NP. 
	
\end{proof}

\bibliographystyle{named}
\bibliography{refer}

\end{document}

%% file: macros.tex
\newtheorem{theorem}{Theorem}[section]
\newtheorem{corollary}[theorem]{Corollary}

\newtheorem{lemma}[theorem]{Lemma}
\newtheorem{proposition}[theorem]{Proposition}
\newtheorem{claim}[theorem]{Claim}

\newtheorem{example}[theorem]{Example}

\newtheorem*{conjecture*}{Conjecture}
\newtheoremstyle{nonindented}{1ex}{1ex}{}{}{\bfseries}{.}{.5em}{}
\newtheoremstyle{indented}{1ex}{1ex}{\itshape\addtolength{\leftskip}{0.6cm}\addtolength{\rightskip}{0.6cm}}{}{\bfseries}{.}{.5em}{}
\theoremstyle{nonindented}
\theoremstyle{indented}
\theoremstyle{plain}

\newcommand{\set}[1]{\left\{ #1 \right\}}
\newcommand{\union}{\cup}

\newcommand{\floor}[1]{\lfloor {#1} \rfloor}

\renewcommand{\tilde}{\widetilde}
\renewcommand{\bar}{\overline}

\DeclareMathOperator{\poly}{poly}

\def\min{\qopname\relax n{min}}
\def\max{\qopname\relax n{max}}

\def\Pr{\qopname\relax n{\mathbf{Pr}}}
\def\Ex{\qopname\relax n{\mathbf{E}}}

\newcommand{\RR}{\mathbb{R}}
\newcommand{\RRp}{\RR_+}

\newcommand{\ZZ}{\mathbb{Z}}

\def\A{\mathcal{A}}

\def\D{\mathcal{D}}

\def\F{\mathcal{F}}

\def\I{\mathcal{I}}

\def\P{\mathcal{P}}

\def\O{\mathcal{O}}

\def\eps{\epsilon}
\def\sse{\subseteq}

\newcommand{\INPUT}{\item[\textbf{Input:}]}
\newcommand{\OUTPUT}{\item[\textbf{Output:}]}
\newcommand{\PARAMETER}{\item[\textbf{Parameter:}]}

\newcommand{\mini}[1]{\mbox{minimize} & {#1} &\\}
\newcommand{\maxi}[1]{\mbox{maximize} & {#1 } & \\}

\newcommand{\st}{\mbox{subject to} }
\newcommand{\con}[1]{&#1 & \\}
\newcommand{\qcon}[2]{&#1, & \mbox{for } #2.  \\}
\newenvironment{lp}{\begin{equation}  \begin{array}{lll}}{\end{array}\end{equation}}
\newenvironment{lp*}{\begin{equation*}  \begin{array}{lll}}{\end{array}\end{equation*}}